%% file: arXiv_csma_double_final.tex
\documentclass[journal]{IEEEtran}
\IEEEoverridecommandlockouts
\usepackage{dsfont}
\usepackage{pdfsync}
\usepackage{bbm}
\usepackage{cite}
\usepackage{mathrsfs}
\usepackage{supertabular}
\usepackage{color}
\usepackage{graphics} % for pdf, bitmapped graphics files
\usepackage{latexsym, amsmath,theorem, amsfonts, amssymb,graphicx,array,tabularx,algorithm}
\usepackage{subfigure}
\usepackage{paralist}
\usepackage{caption}
\usepackage{algpseudocode}
\usepackage{tikz}
\usepackage{pgfplots}
\usepgfplotslibrary{groupplots}
\usepackage{booktabs}
\usetikzlibrary{external}
\newlength\figureheight
\newlength\figurewidth

\pgfplotsset{every axis/.append style={font=\scriptsize}}
\usetikzlibrary{calc,plotmarks}

\theoremstyle{theorem}
\newtheorem{theorem}{Theorem}
\newtheorem{lemma}{Lemma}

\newtheorem{proposition}{Proposition}
\newtheorem{remark}{Remark}

\newtheorem{problem}{Problem}

\newcommand{\R}{{\rm  I\kern-2pt R}}
\renewcommand{\Re}{{\rm  I\kern-2pt R}}

\newcommand{\rank}{{\mathrm{rank}}}

\newcommand{\ie}{{\it i.e.}}

\newcommand{\tr}[1]{\mathrm{Tr}\left[#1\right]}
\newcommand{\E}[1]{\mathbb E\left[#1\right]}
\newcommand{\bb}[1]{\textcolor[rgb]{0.00,0.00,0.00}{#1}}

\newcommand{\I}{{\rm I}}
\newcommand{\dd}{{\rm d}}

\newcommand{\T}{^\top}
\newcommand{\OP}{\overline{P}}

\newcommand{\p}[1]{\Pr\left\{#1\right\}}
\newcommand{\ap}{\widehat{\alpha}_{i*}^+}
\newcommand{\an}{\widehat{\alpha}_{(i+1)*}^-}
\newcommand{\apn}{\widehat{\alpha}_{(n-1)*}^+}
\newcommand{\ann}{\widehat{\alpha}_{n*}^-}

\title{On Stochastic Sensor Network Scheduling for Multiple Processes}
\author{Duo Han$^*$, Junfeng Wu$^\dagger$, Yilin Mo$^*$, Lihua Xie$^*$
\thanks{This research is partially funded by the Republic of Singapore National Research Foundation
(NRF) through a grant to the Berkeley Education Alliance for Research in Singapore (BEARS)
for the Singapore-Berkeley Building Efficiency and Sustainability in the Tropics (SinBerBEST)
Program. BEARS has been established by the University of California, Berkeley as a center for
intellectual excellence in research and education in Singapore.}
\thanks{$*$: School of Electrical and Electronic Engineering, Nanyang Technological University, Singapore, 639798. E-mail: {dhanaa,ylmo,elhxie@ntu.edu.sg}.}
\thanks{$\dagger$: ACCESS Linnaeus Center, School of Electrical Engineering, KTH Royal Institute of Technology, Stockholm, Sweden. E-mail: {junfengw@kth.se}.}
}

\begin{document}
%\begin{frontmatter}
\maketitle
\begin{abstract}
We consider the problem of multiple sensor scheduling for remote state estimation of multiple process over a shared link. In this problem, a set of sensors monitor mutually independent dynamical systems in parallel but only one sensor can access the shared channel at each time to transmit the data packet to the estimator. We propose a stochastic event-based sensor scheduling in which each sensor makes transmission decisions based on both channel accessibility and distributed event-triggering conditions. The corresponding minimum mean squared error (MMSE) estimator is explicitly given. Considering information patterns accessed by sensor schedulers, time-based ones can be treated as a special case of the
proposed one. By ultilizing realtime information,
{the proposed schedule outperforms the time-based ones in terms of the estimation quality.} Resorting to solving an Markov decision process (MDP) problem with average cost criterion, we can find optimal parameters for the proposed schedule. As for practical use, a greedy algorithm is devised for parameter design, which has rather low computational complexity. We also provide a method to quantify the performance gap between the schedule optimized via MDP and any other schedules.
\end{abstract}
%\end{frontmatter}

\section{Introduction}\label{section:introduction}
Sensor scheduling is crucial for remote state estimation in cyber-physical systems (CPS). Typically, the main task for sensor scheduling is to improve estimation quality subject to communication constraints \cite{vitus2012efficient,wu2013event,shi2013optimal,liu2014optimal,zhao2014optimal,mo2014infinite,han2015stochastic,shi2016event}. In this paper we focus on the problem of bandwidth constrained sensor scheduling for state estimation. To be specific, the distributed sensors in charge of different monitoring tasks are sharing a common channel for data transmission. At each time instant, only one sensor can access the communication channel. We consider a sensor schedule deciding which sensor is able to access the channel at each time in order to optimize the overall state estimation quality.

The problem of single sensor scheduling subject to limited transmission rate has been studied in  \cite{yang2011deterministic,wu2013event,zhao2014optimal,han2015stochastic}. Generally speaking, if a sensor uses only prior knowledge of systems, an optimal policy for transmission is very likely to transmit the data packets periodically \cite{yang2011deterministic,shi2013optimal,zhao2014optimal}.
This kind of policies are referred to as time-based schedules.
When the sensor makes transmission decisions based on realtime innovations, which generally has better estimation performance then periodic ones if properly designed but induces higher computational complexity, the transmission is likely to be random \cite{wu2013event,han2015stochastic,shi2016event}.
This kind of policies are referred as event-based schedules.
Weerakkody et al. \cite{sean2015} extended \cite{han2015stochastic} by considering multiple sensors monitoring one process, but without any constraint on channel accessibility.

\bb{Compared with the single system case, not enough research
efforts have been put in the case of different sensors
monitoring different systems, which is widely encountered
in practice. A simple motivating example is the
underground petroleum storage using WirelessHART
technology in \cite{song2008wirelesshart}. Underground
salt caverns are often used for crude oil storage.
Brine and crude oil flowing in both directions are measured
by sensors and reported to the
control center through a gateway device. The control
center aims to maintain a certain pressure inside
the caverns. Obviously, each sensor competes with
others for the gateway access to achieve their own goal.
Therefore, a schedule for optimizing an objective function taking the benefits of all sensors into account is desirable. A few preliminary works have considered this case.} Savage and La Scala \cite{savage2009optimal} studied multiple sensor scheduling for a set of scalar Gauss-Markov systems over a finite horizon. They considered the optimality in terms of terminal estimation error covariance. An optimal policy is to schedule all the transmissions in the end of the horizon. However, the terminal error covariance metric is only suitable for finite horizon scheduling problems. To study an infinite horizon scheduling, Shi et al. \cite{shi2012scheduling} adopted a metric of averaged estimation error, studied two multi-dimensional systems over an infinite horizon and proposed an explicit optimal periodic sensor schedule. The conclusion in \cite{shi2012scheduling} benefits from the mutual exclusiveness of two sensors. When the number of sensors is beyond three, closed-form optimal scheduling is formidable to obtain. Another similar problem where one sensor is scheduled to monitor multiple processes was studied in \cite{lin2013scheduling}. The problem is same with the multiple sensor scheduling problem in nature. The authors proposed algorithms to search schedules such that the error covariance of each system is bounded by some constant matrix. Unlike the single-sensor case, multi-sensor event-based scheduling is not fully investigated to push the limit of performance. A preliminary work by Han et al. \cite{han2013event} improved \cite{shi2012scheduling} by designing an event-based schedule which depends on the importance of a sensor's measurements.

In this note we consider an event-based sensor scheduling design for a set of sensors. At each time every sensor makes a transmission decision based on both channel accessibility by sensing the carrier and the importance of its own data by checking some criteria. Only when the channel is accessible and the data is considered as being sufficiently important, the sensor will
transmit the data packet. The sensor scheduling studied in this note is restricted to a kind of stochastic event-based sensor scheduling since it can maintain Gaussianity of estimation process and bypass the nonlinear problem induced by the event-triggering mechanism, e.g., \cite{wu2013event,shi2016event}.
Compared with time-based schedules, the proposed event-based schedule dynamically assigns the communication resource according to the needs of sensors.
From operating principle point of view, the time- and event-based
sensor schedules are analogous to time division multiple access (TDMA) and carrier sense multiple access with collision avoidance (CSMA/CA) medium access control (MAC) in communication networks, respectively.
The main contributions are summarized as follows: (1) We first propose an event-based scheduling infrastructure with network cooperation and self event-triggering mechanism. Any time-based schedules can be treated as a special case of the framework. (2) Based on the underlying schedule, we derive the minimum mean squared error (MMSE) estimator and analyze the communication behaviour of each sensor. (3) We model an Markov decision process (MDP) problem with average cost criterion to seek the optimal parameters for the class of proposed stochastic schedules. (4) For computational simplicity, we also propose a greedy parameter design method. Moreover, we are able to analytically quantify the performance any schedule by showing a lower bound of the optimal cost.

%The remainder of the paper is organized as follows. Section~\ref{section:problem-setup} describes the problem of interest. MMSE estimation under the proposed schedule is investigated in Section \ref{section:estimation}. How to optimally design the parameter is presented in Section \ref{section:parameter_design}. More results on suboptimal schedules are also discussed. Concluding remarks are given in Section \ref{section:conclusion}.
%

\textit{Notations}: $\mathbb Z_+$ is the set of positive integers. $\mathbb{S}_{+}^{n}$ ($\mathbb{S}_{++}^{n}$) is the set of $n$ by $n$ positive semi-definite (definite) matrices. $\mathrm{Tr}[\cdot]$ denotes the trace of a matrix. $X^{1/2}$ denotes the square root of $X\in\mathbb{S}_{+}^{n}$. For functions $f$ with appropriate domains, $f^{0}(X) := X$, and $f^{t}(X) := f\big(f^{t-1}(X)\big)$. $x[i]$ represents the $i$th entry of the vector $x$. For a matrix $X$, $X(i,j)$ represents the entry on the $i$th row and $j$th column of $X$. For $x\in\mathbb R$, $\lfloor x\rfloor$ is the largest integer that is not larger than $x$ and
$\lceil x\rceil$ is the smallest one that is not less than $x$. Denote the set or sequence $\{x_i\}_{i=j}^k=\{x_j,x_{j+1},\ldots,x_k\},~j\leq k$ and if $j>k$ then $\{x_i\}_{i=j}^k=\emptyset$.

\section{Problem Setup} \label{section:problem-setup}
\subsection{System Model}

Consider the following $n$ mutually independent linear time-invariant (LTI) systems, which are monitored by $n$ sensors respectively:
\begin{subequations}
\begin{align}
x_{i}(k+1) & =  {A}_ix_{i}(k) + \omega_i(k),\label{eq:system-dynamics} \\
y_i(k) & = {C}_i x_i(k) + \nu_i(k),~ i\in\mathcal S\label{eq:sensor-dynamics}
\end{align}
\end{subequations}
where $\mathcal S:=\{1,\ldots,n\}$ is the index set of the $n$ processes or sensors,
$x_i(k) \in \mathbb{R}^{n_i}$ is the state of the $i$th process at time $k$ and $y_i(k) \in \mathbb{R}^{m_i}$ is the measurement obtained by the $i$th sensor at time $k$.
For shorthand denote $s_i$ as the $i$th sensor.
The system noise $\omega_i(k)$'s, the measurement noise $\nu_i(k)$'s and the initial system state $x_i(0)$ are mutually independent zero-mean Gaussian random variables with covariances ${Q}_i > 0$, ${R}_i > 0$, and ${\Pi}_i\geq 0$, respectively. Assume that $({A}_i, {C}_i)$ is \bb{detectable}. Furthermore, we assume that ${A}_i$ is unstable like \cite{shi2012scheduling} for two reasons: (1) unstable systems bring out stability issues rather than stable ones do; (2) most process estimation tasks will become unpredictable if left unattended too long.

Each sensor measures its corresponding system state and generates a local estimate first. More specifically, at each time $s_i$ locally  runs a Kalman filter to compute the minimum mean square error (MMSE) estimate of $x_i(k)$, \ie, $\hat{x}_{i,{\rm local}}(k) := \mathbb{E}[ x_i(k)| \{y_i(l)\}_{l=0}^k]$. The corresponding estimation error and error covariance are
\begin{align}
\epsilon_{i,{\rm local}}(k) &:= x_i(k)-\hat{x}_{i,{\rm local}}(k),\label{eq:epsilon_def}\\
    {P}_{i,{\rm local}}(k) &:= \mathbb{E}\left[\epsilon_{i,{\rm local}}(k)(\epsilon_{i,{\rm local}}(k))^{\T}|\{y_i(l)\}_{l=0}^k\right].
\end{align}
The quantities $\hat{x}_{i,{\rm local}}(k)$ and ${P}_{i,{\rm local}}(k)$ can be obtained through a standard Kalman filter \cite{anderson79}. It is well known that ${P}_{i,{\rm local}}(k)$ converges exponentially fast to $\OP_i\in\mathbb S_+^{n_i}$ which is the solution of a discrete algebraic Riccati equation \cite{anderson79}. Since we consider a problem over the infinite horizon in the sequel, we ignore the transient behaviour and make a standing assumption that ${P}_{i,{\rm local}}(k) = \OP_i,\forall k\in \mathbb Z_+$.

For the purpose of state estimation, all sensors transmit their own local estimates to the remote estimator over a shared channel. In this work we consider a bandwidth-limited sensor network which allows one sensor to access the channel each time instant. In other words, \emph{only a single sensor can send its estimate over the shared channel at each time instant while the others still keep their local copies}.

\bb{
\subsection{Scheduling Problem}
The problem of interest is \emph{how to efficiently schedule the transmission of sensors at each time in terms of some performance metric}.}

\bb{
We first define some mathematical notations. Denote the transmission indicator for $s_i$ at time $k$ as a binary variable $\gamma_i(k)$, \ie, $\gamma_i(k)=1$ means $s_i$ sends data and vice versa. A sensor schedule $\theta$ is defined as a sequence of $\gamma_i(k)$, \ie, $$\theta:=\{\gamma_i(k)\}_{i\in\mathcal S,k\in \mathbb Z_+}.$$
Moreover, we define a collection of the information sets $\mathcal I_{i,\theta}(k)$'s for $i\in\mathcal S$ at the estimator side as
\begin{align}
\mathcal I_{i,\theta}(k) :=
\{\gamma_i(0)\hat{x}_{i,{\rm local}}(0), \ldots, \gamma_i(k)\hat{x}_{i,{\rm local}}(k)\}. \label{eq:25}
\end{align}
Due to the mutual independency among the systems, the estimator computes the estimate of $x_i(k)$ as $\hat x_i(k):=\mathbb{E}[ x_i(k)|\mathcal I_{i,\theta}(k)]$ with the corresponding error covariance $P_i(k):= \mathbb{E}[(x_i(k)-\hat{x}_i(k))( x_i(k)-\hat{x}_i(k))^{\T}|\mathcal I_{i,\theta}(k)]$. Note that $\gamma_i(k),\hat x_i(k)$ and $P_i(k)$ are functions of $\theta$ though we do not explicitly show that in the notations.
}

\bb{
Similar to \cite{shi2012scheduling}, we use the overall average expected estimation error covariance as a performance metric, \ie,
\begin{align}
J(\theta) := \limsup_{T\rightarrow\infty} \frac{1}{T}\sum_{i\in\mathcal S} \tr{\sum_{k=0}^{T-1} \left(\E{P_i(k)}\right)}.\label{eq:cost-function-total}
\end{align}
Let $J_i(\theta)$ denote the individual expected estimation error covariance correspondingly. The problem of interest is formally stated as:
\begin{align}
& \underset{\theta}{\textit{minimize}}
& &J(\theta)
& \textit{subject to} &&\sum_{i\in\mathcal S}\gamma_i(k) = 1.
\end{align}
}
% where the cost function $J_i(\theta)$ over an infinite horizon for estimating the $i$th process is defined as follows:
%\begin{equation}\label{eq:cost-function-definition}
%J_i(\theta) := \limsup_{T\rightarrow\infty} \frac{1}{T}\tr{\sum_{k=0}^{T-1} \left(\E{P_i(k)}\right)}.
%\end{equation}
%We are interested in the following optimization problem\footnote{
%Notice that the optimality must be reached at $\sum_{i\in\mathcal S}\gamma_i(k)=1$ since any transmission is always preferred to convey more information.}:
%\begin{problem} \label{problem:main-problem-1}
%\begin{align}
%& \underset{\theta\in\Theta }{\textit{minimize}}
%& &J(\theta)=\sum_{i\in\mathcal S}J_i(\theta), \label{eq:15}\\
%& \textit{subject to} &&\sum_{i\in\mathcal S}\gamma_i(k) \leq 1,~ \forall k\in\mathbb Z_+.\label{eq:1}
%\end{align}
%\end{problem}
\bb{
Time-based schedules refer to that $\theta$ is a function of time only. In that case, $\gamma_i(k)$ is independent of $\hat{x}_{i,{\rm local}}(k)$. The optimal scheduling solution turns out to be a periodic TDMA-like schedule like \cite{shi2012scheduling}. The main advantage of this type of schedule is its simplicity. From \eqref{eq:25}, however, we can see that when $\gamma_i(k)=0$, $\mathcal I_{i,\theta}(k)$ contains no information on $\hat{x}_{i,{\rm local}}(k)$ and thus the estimator gains nothing about the system state from $s_i$ at time $k$ when $\gamma_i(k)=0$.
}

\bb{
To outperform the time-based schedule, we study a type of event-based schedule meaning that $\gamma_i(k)$ is a function of the real-time estimates or measurements which contain information on the underlying system state. Both the sensor and the estimator know how $\gamma_i(k)$ at each time is determined based on some triggering rules. Hence, even with $\gamma_i(k)=0$, the estimator can still infer some information about the system. Therefore, the event-triggering mechanism improves the estimation performance by providing extra information.
}

\subsection{Stochastic Event-based Sensor Schedule}
\setlength{\figureheight}{2.3cm}
\setlength{\figurewidth}{6.5cm}
\begin{figure}
  \centering
  %\inputtikz{realization}
  \includegraphics[width=2.7in,height=1in]{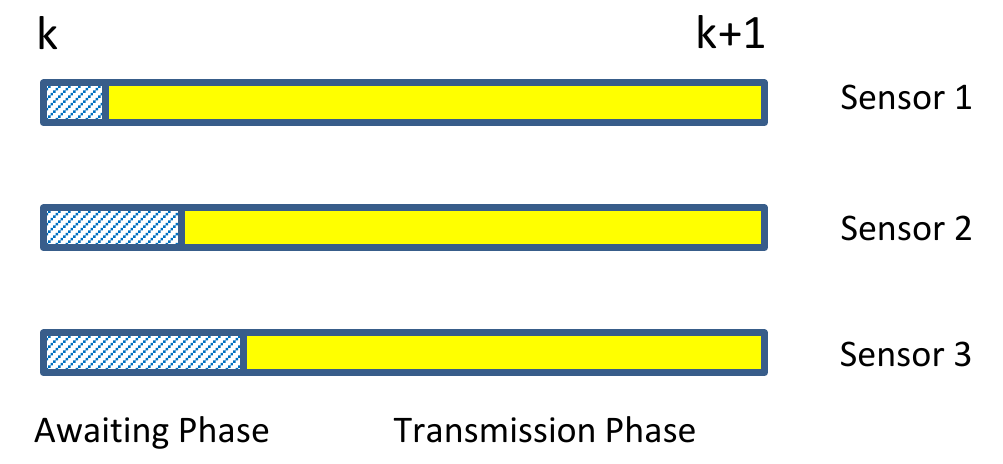}
  \caption{Illustration of the awaiting phase and transmission phase for three sensors between time $k$ and $k+1$.}
  \vspace{-5mm}
  \label{fig:csma}
\end{figure}

Mimicking the protocol of Carrier Sense Multiple Access with Collision Avoidance (CSMA/CA) widely used in wireless sensor networks, we propose a distributed event-based sensor schedule to enhance the overall estimation performance compared with the time-based schedules. Now we introduce the infrastructure. There are two phases for each sensor in each transmission frame: the awaiting phase and the transmission phase. Since the packet transmission time is often much larger than the uncertainties in communication networks such as the propagation delays, it is safe to assume the transmission time is much larger than the awaiting time during each epoch. Any sensor listens to the channel carrier for a short period before it sends anything. If the channel is occupied during the awaiting phase, then the sensor holds its data; otherwise, it sends the packet to the estimator in the transmission phase. The ends of awaiting phase of all sensors are made different in order to avoid collision (see Fig. \ref{fig:csma}). In other words, the sensors form a queue $q(k)$ with different queueing time in the awaiting phase. For example, the queue $(s_1,\ldots,s_{n})$ means $s_i$ has higher priority to access the channel than $s_{i+1}$. The queue $q(k)$ can be time varying or invariant and denote $\mathcal Q$ to be the set of possible priority queues, \ie, all permutations of $s_1,\ldots,s_{n}$.

The idea of event-based scheduling behind is as follows. If the data of $s_i$ contains little innovative information, which can be checked by some criteria introduced later, then $s_i$ will be unlikely to transmit the data and $s_{i+1}$ in the queue can take the chance to use the channel. The queue implemented on the top of carrier sensing is useful here to resolve the conflict when more than one sensor want to transmit. In other words, the transmission decision of $s_i$ depends on the importance of local data and the channel accessibility.

%\begin{remark}
%In the traditional CSMA/CA protocol, the sensors usually wait for a random period of time if they detect the channel is busy. One of the benefits of such random backoff is to get rid of the problem  of sensor synchronization. But at the same time the overall performance quality will decrease since the most important data cannot be guaranteed to transmit. Unlike CSMA/CA, the proposed schedule needs sensor synchronization to form a queue for all sensors. As a reward, the sensory data of more importance will take a higher priority to access the channel, and thus improve the overall estimation quality. Moreover, the sensors behind in the queue also have a chance to transmit if the event of the sensors ahead is not triggered. The proposed schedule can be extended into an asynchronous version which is left as a future work.
%\end{remark}

%\begin{remark}
%Unlike CSMA/CA, the proposed schedule needs sensor synchronization to form a queue for all sensors. As a reward, the sensory data of more importance will take a higher priority to access the channel, and thus improve the overall estimation quality.
%\end{remark}

Let us first define two binary indicators before formally proposing the schedule: the data importance indicator $\eta_i(k)$ and the channel accessibility indicator $\mu_i(k)$. Various event-triggering criteria for determining the importance of a single measurement have been investigated in \cite{wu2013event,han2015stochastic,sean2015}. In the spirit of \cite{han2015stochastic,sean2015}, we design a stochastic event-triggering rule to check the importance of the data. The dominant feature of the stochastic event-triggering mechanism is to maintain the Gaussianity of the estimation process and bypass the nonlinear problem such as \cite{wu2013event}. We define $\eta_i(k)$ based on \textit{the difference between the MMSE estimate under $\theta$ at local sensors and the prediction at the estimator side}, \ie,
$$\epsilon_i(k) := \hat{x}_{i,{\rm local}}(k) - A_i\hat{x}_i(k-1),$$ and the corresponding covariance as
\begin{align}
\Sigma_i(k) := \E{(\epsilon_i(k))\T\epsilon_i(k)|\mathcal I_{i,\theta}(k-1)}.\label{eq:sigma_def}
\end{align}
Define the mapping $\phi_i(z,\Pi):\mathbb R^{n_i}\times \mathbb S^{n_i}_{+}\mapsto [0,1]$ as
\[\phi_i(z,\Pi) := \exp\left(-\frac{1}{2}z\T \Pi^{-1}z\right)\footnote{We abuse the notation $X^{-1}$ to represent the inverse of $X\in \mathbb S^{n_i}_{++}$ or the Moore-Penrose pseudo-inverse of a singular $X \in \mathbb S^{n_i}_{+}$. Similarly, we use $\det(X)$ to denote the determinant of a square matrix $X$ or its pseudo-determinant if $X$ is singular.}.\]
Now we define the data importance indicator $\eta_i(k)$ as follows:
\begin{align}\label{eq:eta}
\eta_i(k) := \left\{
                       \begin{array}{ll}
                         1, & \hbox{if } \xi_i(k) > \phi_i(\epsilon_i(k),\alpha_i(k)\Sigma_i(k)),\\
                         0, & \hbox{otherwise.}
                       \end{array}
                     \right.
\end{align}
\bb{where $\xi_i(k)\sim U[0,1]$ is an i.i.d. auxiliary random variable, and $\alpha_i(k)$ is a tunable parameter which reflects the importance of the data packet. For smaller $\alpha_i(k)$, $\eta_i(k)$ is more likely to be $1$. Namely, for a more important sensor, we would like to pick a smaller $\alpha_i(k)$ to ensure it can transmit more data.}

\bb{The tuning parameter $\{\alpha_i(k)\geq 0\}_{i\in\mathcal S,k\in\mathbb Z_+}$ and the queue parameter $\{q(k)\}_{k\in\mathbb Z_+}$, depending on the historical arrival pattern of all sensors, are designed and broadcasted by the central estimator. The estimator usually has advantageous resources compared with sensors, \ie, stronger computation capability and larger energy storage. The design procedure will be discussed in Section \ref{section:parameter_design}.}

Let $\mu_i(k)$ denote the channel accessibility indicator \footnote{For concise presentation, by relabelling the sensors we assume the queue $q(k)=(s_1,\ldots,s_{n})$ by default. More discussion on the design of the optimal queue is shown in Section \ref{section:parameter_design}.}, \ie,
\begin{align}\label{eq:mu}
\mu_i(k) := \left\{
                       \begin{array}{ll}
                         1, & \hbox{if } \sum_{j=1}^{i-1}\gamma_j(k) = 0\\
                         0, & \hbox{otherwise}
                       \end{array}
                     \right..
\end{align}

Now we are ready to propose an event-based schedule $\theta_+:=\{\gamma_i(k)\}_{i\in\mathcal S,k\in\mathbb Z_+}$ and denote $\Theta_+$ to be the set of all event-based schedules in \eqref{eq:event} subject to all possible $\{\alpha_i(k)\}_{i\in\mathcal S,k\in\mathbb Z_+}$ and $\{q(k)\}_{k\in\mathbb Z_+}$:
\begin{align}\label{eq:event}
\gamma_i(k) := \left\{
                       \begin{array}{ll}
                         \mu_i(k)\eta_i(k), & \hbox{if } i\neq n\\
                         \mu_i(k), & \hbox{otherwise}
                       \end{array}
                     \right..
\end{align}
Notice that $\gamma_i(k)$ depends on $\epsilon_i(k)$ and $\Sigma_i(k)$, and $\epsilon_i(k+1)$ and $\Sigma_i(k+1)$ depends on $\gamma_i(k)$. Thus $\theta_+$ cannot be determined offline. Also note that if $s_n$ with the last priority detects that the channel is idle, it sends the data without checking \eqref{eq:eta}. The condition \eqref{eq:eta} implies that if the prediction error is small, $\eta_i(k)$ is very likely to be $0$ which prevents unecessary transmission.

\bb{In the subsequent sections, we analyze the estimation performance using such a schedule and design the parameters $\{\alpha_i(k)\geq 0\}_{i\in\mathcal S,k\in\mathbb Z_+}$ and the queue parameters $\{q(k)\}_{k\in\mathbb Z_+}$ optimally by formulating optimization problems.}
\section{Optimal Filtering}\label{section:estimation}

\bb{
After proposing the stochastic event-based schedule $\theta_+$, we are keen in deriving the MMSE estimator under $\theta_+$ and  quantifying the estimation performance of $\theta_+$. Moreover, from the viewpoint of communication, we present the transmission probability of each sensor.
}

To facilitate derivations, we define the following operators for $X\in \mathbb S_+^{n_i}$ and $\alpha \in\mathbb R$:
\begin{align}
h_i(X) &:= A_iXA_i\T + Q_i,\\
g_i(X,\alpha) &:= \frac{\alpha}{1+\alpha}(A_iXA_i\T + h_i(\OP_i)-\OP_i),\\
t_i(X,\alpha) &:= \frac{1}{1+\alpha}\OP_i + \frac{\alpha}{1+\alpha}h_i(X).
\end{align}
The following lemma on the properties of the innovation and estimation error is useful for proving the main result whose partial proof can be found in \cite{wu2014stochastic}. Let the incremental innovation for $s_i$ be denoted as
\begin{align}
\delta_i(k):=\hat{x}_{i,{\rm local}}(k)-A_i\hat{x}_{i,{\rm local}}(k-1).\label{eq:delta_def}
\end{align}
\begin{lemma}\label{lemma:1}
Given $\epsilon_{i,{\rm local}}(k),\delta_i(k)$ defined in \eqref{eq:epsilon_def} and \eqref{eq:delta_def}, the following statistical properties hold:
\begin{enumerate}[(i)]
\item $\delta_i(k)$ is zero mean Gaussian distributed, \ie, $\delta_i(k)\sim \mathcal N(0,h_i(\OP_i)-\OP_i)$.
\item $\E{\delta_i(k)(\delta_i(j))\T}=0$ for any $k\neq j$.
\item $\E{\epsilon_{i,{\rm local}}(k)(\delta_i(k_0))\T}=0$ for any $k_0\leq k$.
\item $\E{\hat{x}_{i,{\rm local}}(k)(\epsilon_{i,{\rm local}}(k_1))\T} = 0$ for any $k_1\geq k$ and $\E{\hat{x}_{i,{\rm local}}(k)(\delta_i(k_2))\T} = 0$ for any $k_2> k$.
\end{enumerate}
\end{lemma}

The following lemma on Bayes' inference is presented before showing the main result.
\begin{lemma}\label{lemma:2}
Suppose $z\in\mathbb R^{{n_i}}$ is a Gaussian random variable with $z\sim \mathcal N(0,Z)$ with
respect to the Lebesgue measure on $\mathcal Z:=\{Z^{1/2}x:x\in\mathbb R^{{n_i}}\}$ and $\xi$ is uniformly distributed over $[0,1]$. The following statements hold:
\begin{enumerate}[(i)]
\item The occurring probability of the following event is $\p{\xi\leq \phi_i(z,\Pi)} =\det(\I+Z\Pi^{-1})^{-1/2}$.
\item The conditional pdf of $z$ is $f(z|\xi\leq \phi_i(z,\Pi))\sim \mathcal N\left(0,(Z^{-1}+\Pi^{-1})^{-1}\right).$
\end{enumerate}
\end{lemma}
\begin{proof}
First we have
\begin{align}
&\p{\xi\leq \phi_i(z,\Pi)}  = \E{\exp\left(-\frac{1}{2}z\T \Pi^{-1}z\right)}\notag \\
&=\int_{\mathcal Z} \frac{(2\pi)^{-r/2}}{\det(Z)^{1/2}}\exp\left(-\frac{1}{2}z\T Zz\right)\exp\left(-\frac{1}{2}z\T \Pi^{-1}z\right)\dd z\notag\\
&= \det(\I+Z\Pi^{-1})^{-\frac{1}{2}},\label{eq:24}
\end{align}
where $r=\rank(Z)$. Statement (ii) is a direct result of the Bayes' theorem.
\end{proof}

For notational simplification, denote $
\tilde{\gamma}_{i}(k) := 1-\gamma_i(k)$, $\tilde{\mu}_i(k):=1-\mu_i(k).$
Furthermore, denote the leave duration of $s_i$ as $\tau_i(k):=\min\{k-k_0:\gamma_i(k_0)=1,k_0 \leq k\}$. Now we are ready to present the main result.

\begin{theorem}\label{thm:optimalestimator}
Under the proposed schedule $\theta_+$ given in \eqref{eq:event}, the MMSE state estimate for each process is given by
\begin{align}\label{eq:2}
\hat x_i(k) = \left\{
                \begin{array}{ll}
                  \hat{x}_{i,{\rm local}}(k), & \hbox{if } \gamma_i(k) = 1\\
                  A_i\hat{x}_i(k-1), & \hbox{if } \gamma_i(k) = 0
                \end{array}
              \right.,
\end{align}
and the corresponding estimation error covariance is
\begin{align}\label{eq:3}
P_i(k) = \left\{
           \begin{array}{ll}
             \OP_i, & \hbox{if } \gamma_i(k) = 1 \\
             h_{i}(P_i(k-1)), & \hbox{if } \mu_i(k) = 0 \\
             t_i(P_i(k-1), \alpha_i(k)), & \textrm{otherwise}
           \end{array}
         \right..
\end{align}
Moreover, $\Sigma_i(k)$ in \eqref{eq:sigma_def} is given by
\begin{align}
\Sigma_i(k) = h_i(P_i(k-1))-\OP_i.\label{eq:sigma}
\end{align}
\end{theorem}

\begin{proof}
When $\gamma_i(k)=1$, from \cite{anderson79} we know that
$\E{x_i(k)|\hat x_{i,{\rm local}}(k),\mathcal I_{i,\theta}(k-1)} = \hat x_{i,{\rm local}}(k),$
and $P_i(k)=\OP_i$.

When $\mu_i(k)=0$ ($\gamma_i(k)$ must be $0$) which implies the channel is occupied by another sensor, the estimator can only do prediction on the estimate of $x_{i}(k)$, \ie
\begin{align}
\hat x_i(k) = A_i\hat{x}_i(k-1),~P_i(k) = h_{i}(P_i(k-1)).\label{eq:9}
\end{align}

The two cases above are easy to analyze. Next we consider the remaining case, \ie, when $\mu_i(k)\tilde{\gamma}_{i}(k)=1$. The following equation is easy to verify and useful for subsequent derivations.
\begin{align}
x_i(k)=A_i\hat{x}_i(k-1)+\epsilon_{i,{\rm local}}(k)+\epsilon_i(k).\label{eq:4}
\end{align}

Without loss of generality, assume that $\gamma_i(k)=1$ or $\mu_i(j)=0$ for any $j\in[k-\tau_i(k),k]$, and $\mu_i(k+1)\tilde{\gamma}_i(k+1)=1$ occurs at time $k+1$. Moreover, the following recursive equation always holds:
\begin{align}
\epsilon_i(k) =A_i\epsilon_i(k-1) +\delta_i(k).\label{eq:10}
\end{align}
Also from \eqref{eq:9} and \eqref{eq:4} we have $x_i(k+1)=A_i^{\tau_i(k+1)}\hat{x}_{i,{\rm local}}(k-\tau_i(k+1))+\epsilon_{i,{\rm local}}(k+1)+\epsilon_i(k+1) = A_i^{\tau_i(k+1)}\hat{x}_{i,{\rm local}}(k-\tau_i(k+1))+\epsilon_{i,{\rm local}}(k+1)+\sum_{j=0}^{\tau_i(k)}A_i^j\delta_i(k+1-j)$,  where the last two terms on the RHS are mutually independent in view of Lemma \ref{lemma:1}(iii). Due to Lemma \ref{lemma:1}(i) and the fact \cite{anderson79} that $
f(\epsilon_{i,{\rm local}}(k+1)|\mathcal I_{i,\theta}(k+1))=f(\epsilon_{i,{\rm local}}(k+1))\sim \mathcal N(0,\OP_i)$, we have $\E{x_i(k+1)|\mathcal I_{i,\theta}(k+1)}=A_i\hat{x}_i(k).$ Therefore, from \eqref{eq:4} we know that for $j\in[k-\tau_i(k),k+1]$ the equality holds:
\begin{align}
P_i(k) = \OP_i + \Lambda_i(k),\label{eq:11}
\end{align}
where $\Lambda_i(k):=\E{(\epsilon_i(k))\T\epsilon_i(k)|\mathcal I_{i,\theta}(k)}$. Hence we have $\Lambda_i(k) = P_i(k)-\OP_i$.
Then together with \eqref{eq:10} we can conclude that
\begin{multline*}
f(\epsilon_i(k+1)|\mathcal I_{i,\theta}(k)) \sim \mathcal N(0, A_i(P_i(k)-\OP_i)A_i^{\top}+h_i(\OP_i)-\OP_i)
\end{multline*}
which proves \eqref{eq:sigma}.
Then from Lemma \ref{lemma:2}, we have
\begin{align}
f(\epsilon_i(k+1)|\mathcal I_{i,\theta}(k+1))\sim  \mathcal N(0,g_i(P_i(k)-\OP_i,\alpha_i(k))).\label{eq:7}
\end{align}
Thus, from \eqref{eq:11} and \eqref{eq:7} we have
\begin{multline}
f(x_i(k+1)|\mathcal I_{i,\theta}(k+1))\\
\sim \mathcal N(A_i\hat{x}_i(k),\OP_i+g_i(P_i(k)-\OP_i,\alpha_i(k))).\label{eq:8}
\end{multline}
Notice that $\hat{x}_i(k+1)$ is actually $A_i^{\tau_i(k+1)}\hat{x}_{i,{\rm local}}(k-\tau_i(k+1))$ which is the same with the predicted estimate but with a smaller error covariance. Consequently, no matter what $\mu_i(k+2)$ and $\gamma_i(k+2)$ are at time $k+2$, the mutual independence of the three terms in \eqref{eq:4} and the recursion in \eqref{eq:10} hold. Therefore, the conditional pdf of $\epsilon_i(k+2)$ can be computed in a similar fashion as \eqref{eq:7}, which is zero mean Gaussian distributed. Thus $\hat{x}_i(k+2)$ is still a predicted estimate, \ie, $A_i\hat{x}_i(k+1)$, with the corresponding error covariance dependent on $\mu_i(k+2)$ and $\gamma_i(k+2)$. Recursively, we can verify \eqref{eq:2} and \eqref{eq:3} which completes the proof.
\end{proof}

The result in \eqref{eq:3} shows that if $s_i$ is idle due to other competitors, \ie, $\mu_i(k)=0$, then the estimator simply runs a prediction. If $s_i$ decides to hold its packet even if it has the access to the channel, \ie, $\mu_i(k)\tilde{\gamma}_i(k)=1$, the estimator updates the estimate using the information encoded by \eqref{eq:eta}. Using the following lemma, we can see the estimation performance under $\mu_i(k)\tilde{\gamma}_i(k)=1$ is lower bounded by that under $\gamma_i(k)=1$ and upper bounded by that under $\mu_i(k)=0$.

\begin{lemma}\label{lemma:supporting-lemma}
The following statements hold for any $i\in\mathcal S$:
\begin{enumerate}
\item[(i).] For any $l_1,l_2\in\mathbb{Z}_+$ with $l_1 <l_2$, $h_i^{l_1}(\OP_i) \leq h_i^{l_2}(\OP_i)$.
\item[(ii).] For any $l \in \mathbb{Z}_+$, $\tr{{\OP_i}} <\tr{h_i(\OP_i)}\leq\cdots \leq \tr{h_i^{l}(\OP_i)}.$
\item[(iii).] For any $X$ which is a convex combination of $\{h_i^j(\OP_i)\}_{j=0}^\infty$ and any $\alpha\geq 0$, $X\leq t_i(X,\alpha)$ and $\tr{t_i^{l_1}(X,\alpha)}\leq\tr{t_i^{l_2}(X,\alpha)}$, for $l_1 <l_2$.
\end{enumerate}
\end{lemma}
\begin{proof}
%Since $\OP_i\geq 0$ and $Q>0$, we have $h_i(\OP_i)>0$. We write the discrete algebraic Riccati equation into the following form which aligns with the information filter \cite{anderson79}:
%\begin{align*}
% \OP_i = \left(h_i(\OP_i)^{-1}+C_i\T R^{-1} C\right)^{-1}.
%\end{align*}
%It is known that for any positive definite matrices $0<X<Y$, then $X^{-1}>Y^{-1}$. Then we have
%\begin{align}
% h_i(\OP_i)-\OP_i &= \left(h_i(\OP_i)^{-1}\right)^{-1} - \left(h_i(\OP_i)^{-1}+C_i\T R^{-1} C\right)^{-1}>0. \label{eq:18}
%\end{align}
%Then by applying the affine operator $h_i(\cdot)$ on both sides we can prove (i). The statements (ii) and (iii) are direct results from (i).
From \cite[Lemma A.1]{shi2010kalman}, we can prove that $\OP_i \leq h_i(\OP_i)$. Since $h_i(\cdot)$ is affine, we can take $h_i(\cdot)$ on both sides and draw the conclusion (i). Next we show the strictiveness in (ii), \ie, $\tr{{\OP_i}} < \tr{h_i(\OP_i)}$. Assuming $\OP_i=h_i(\OP_i)$, then we can find that $Q=0$ which contradicts the assumption $Q>0$, or $C_i=0$ which contradicts the fact $(A_i,C_i)$ is detectable.  Thus $\OP_i\neq h_i(\OP_i)$. From (i) we immediately have $\tr{{\OP_i}} \leq \tr{h_i(\OP_i)}\leq \cdots \leq \tr{h_i^{l}(\OP_i)}.$ Moreover, we know that
\[
\left\{\begin{array}{ll}
         \OP_i\neq h_i(\OP_i)\\
         \OP_i\leq h_i(\OP_i), \OP_i\in\mathbb S_+^{n_i}
       \end{array}\right.
       \Rightarrow \tr{{\OP_i}} \neq\tr{h_i(\OP_i)}.
\]
The last statement is a direct result of (iii).
\end{proof}

From Lemma \ref{lemma:supporting-lemma}(iii), for any realization of $P_i(k-1)$ we know that $\OP_i\leq t_i(P_i(k-1),\alpha_i(k))\leq h_i(P_i(k-1))$. In other words, compared with the time-based schedules, even if $s_i$ does not transmit anything, the estimator is still likely to obtain an estimate better than a pure prediction.

Denote as $r_i(k)$ the rank of $\Sigma_i(k)$ which can be computed offline from \eqref{eq:sigma}.
The probability of kinds of events during transmission can be computed as follows. For concise notation, denote
\begin{align}
\widehat{\alpha}_i(k)&:=\alpha_i(k)/(1+\alpha_i(k))\label{eq:alpha_def}\\
\beta_i(k)&:=\widehat{\alpha}_i(k)^{r_i(k)/2},~\tilde{\beta}_i(k) := 1 - \beta_i(k).\label{eq:beta_def}
\end{align}
\begin{theorem}\label{thm:pr}
The following statements hold true:
\bb{\begin{align*}
&\p{\eta_i(k)=1}=\tilde{\beta}_i(k),\\
&\p{\mu_i(k)=1|\Xi_i(k)}= \left\{
                       \begin{array}{ll}
                         1, & \hbox{if } i=1\\
                         \prod_{j=1}^{i-1}\beta_j(k), & \hbox{if }i>1
                       \end{array}
                     \right.,\\
&\p{\gamma_i(k)=1|\Xi_i(k)}=\left\{
                       \begin{array}{ll}
                         \prod_{j=1}^{n-1}\beta_j(k) , & \hbox{if } i=n\\
                         \tilde{\beta}_i(k)\prod_{j=1}^{i-1}\beta_j(k) , & \hbox{if }i<n
                       \end{array}
                     \right.,
\end{align*}
where $\Xi_i(k):=\{\eta_{i-1}(k),\ldots,\eta_{1}(k)\}.$}
\end{theorem}
\begin{proof}
From Lemma \ref{lemma:2}(i), we have
$\p{\xi_i(k) < \phi_i(\epsilon_i(k),\alpha_i(k)\Sigma_i(k))}
=[\alpha_i(k)/(1+\alpha_i(k))]^{r_i(k)/2}.$
For (ii), it is easy to see that $\p{\mu_1(k)=1|\Xi_i(k)}=0$ for $s_1$. From \eqref{eq:mu} and \eqref{eq:event}, we know that $\p{\mu_i(k)=1|\Xi_i(k)}=1-\sum_{j=1}^{i-1}\p{\gamma_j(k)=1|\Xi_i(k)}$ and $\p{\gamma_i(k)=1|\Xi_i(k)}=(1-\p{\eta_i(k)=0}|\Xi_i(k))\p{\mu_i(k)=1|\Xi_i(k)}$. After some calculation, omitted for longevity, we obtain (ii) and (iii).
\end{proof}
\begin{remark}
Since $\Sigma_i(k)$ in \eqref{eq:sigma} can be written as $\sum_{j\geq1}\kappa_jh_i^j(\OP_i)-\OP_i$ where $\sum_{j\geq1}\kappa_j=1$, we can guarantee $\Sigma_i(k)$ to be full rank as long as $h_i^j(\OP_i)-\OP_i>0,\forall j\geq1$ holds. A sufficient condition is that $A_i$ has full rank.
\end{remark}
\bb{Since we have known how the error covariance of each process is updated under different conditions of $\{\gamma_i(k),\mu_i(k),\eta_i(k)\}$ and the probability of each outcome in Theorem \ref{thm:pr}, we can next investigate how to optimally set the event-triggers for each sensor.}

\section{Parameter Design}\label{section:parameter_design}
\bb{In this section, we aim to design the optimal parameters for the class of stochastic event-based schedules based on the results in Section \ref{section:estimation}. It turns out that finding the optimal schedule is computationally challenging. Thus, we propose a greedy algorithm to achieve suboptimal schedules. Moreover, we analytically quantify the performance gap between the optimal and suboptimal schedules.}

We first show that the proposed schedule performs at least as good as any time-based schedule.

\begin{theorem}\label{thm:opt}
There exists a nonempty subset of $\Theta_+$ from which $\theta_+$ is at least as good as any time-based schedule $\theta_{\ddagger}$ does, \ie, $\{\theta_+\in\Theta_+: J(\theta_+)\leq J(\theta_{\ddagger}),\forall \theta_{\ddagger}\}\neq \emptyset$.
\end{theorem}
The proof is straightforward since any time-based schedule is a special case of $\theta_+$ in \eqref{eq:event}. For example, if $s_i$ is deterministically scheduled to transmit at time $k$, then one can set $\alpha_i(k)=0$ and $\alpha_j(k)=\infty$ for $j\neq i$. Thus by optimizing the parameters in the schedule $\theta_+$, one can always find a schedule in $\Theta_+$ performs as well as any time-based schedule can.

\subsection{Optimal Parameter Design by Solving an MDP Problem}
Now we are keen on the optimal $\theta_+\in\Theta_+$ of minimizing \eqref{eq:cost-function-total}. By formulating an Markov decision process (MDP) problem with average cost criterion, we can design the optimal time-varying parameters $\{\alpha_i(k)\}_{i\in\mathcal S,k\in\mathbb Z_+}$ and the queue $\{q(k)\}_{k\in\mathbb Z_+}$.

First define the state space $\mathcal V$ to be the set of all possible $(P_1(k),\ldots,P_n(k))$, where $P_i(k)$ belongs to $\mathcal C_i$ which is the set of any convex combination of $h_i^j(\OP_i),\forall j$, \ie, $\mathcal C_i:=\{\sum_j\lambda_jh_i^j(\OP_i)$:$\sum_j\lambda_j=1$\}. The action is a tuple of $b:=(u,q)\in \mathcal U\times \mathcal Q$ where $\mathcal U:=\{u\in \mathbb R^{n}: 0\leq u[i]\leq 1,\forall i\in \mathcal S\}$. The compact action space is denoted as $\mathcal B$, which is identical for each state in $\mathcal V$.  Denote the set $\mathcal K:=\{(v,b):v\in\mathcal V,b\in\mathcal B\}$ which is a Borel space. The transition law $Q(\cdot|v,b)$ with $(v,b)\in \mathcal K$ is a stochastic kernel on $\mathcal V$ given $\mathcal K$, which can be obtained from Theorem \ref{thm:optimalestimator} and Theorem \ref{thm:pr}. The cost function $c(v,b)$ is thus defined as the sum of the trace of each matrix in $v$. We thus model an MDP denoted as $\Omega:=\left(\mathcal V,\mathcal B,Q(\cdot|\cdot,\cdot),c(\cdot,\cdot)\right)$. A decision at time $k$ is a mapping $d(k):\mathcal V\mapsto \mathcal B$.
A policy $\zeta$ for $\Omega$ is a sequence of decision rules.
We define the \emph{average expected cost} of the policy $\zeta$ per unit time by
\begin{align}
g_{\zeta}(v)=\limsup_{T\rightarrow \infty}\frac{1}{T}\mathbb E_{\zeta,v}\left[\sum_{k=0}^{T-1}r(v(k),b(k))\right],\label{eqn:modified_cost}
\end{align}
where $v\in\mathcal V$ is the initial state, and the expectation is taken based on the stochastic process $\{v(k),b(k)\}$ uniquely determined by the policy $\zeta$ and $v$. The target for an MDP problem with average cost criterion is to search an optimal stationary policy $\zeta^*$ for $\Omega$ such that the average cost is minimized, \ie,
$g_{\zeta^*}(v)\leq g_{\zeta}(v),~\forall v\in\mathcal V.$ It is easy to see that finding the optimal policy $\zeta^*$ for $\Omega$ is equivalent to searching the optimal $\theta_+\in\Theta_+$ of minimizing \eqref{eq:cost-function-total}.

\bb{Numerous literature have studied the optimality conditions for a policy of an MDP problem with average cost criterion in Borel spaces such as \cite{hernandez2012further}.
Though some researchers have attempted to solve the MDP problem in Borel spaces like \cite{zhu2005value}, obtaining the optimal cost and the optimal policy in a computationally efficient fashion, however, is generally chanllenging. The technique of discretizing the state space and rebuilding the stochastic kernel from the original model is a popular approach to approximately solve the MDP with Borel spaces. Unfortunately, for solving the discrete MDP, the computational burden is high. For example, the classical policy iteration algorithm for solving an MDP requires the computational effort of $|\mathcal V_d||\mathcal B_d|^2+\frac{1}{3}|\mathcal V_d|$ multiplications/iterations \cite[Chapter 8]{puterman2005markov}, where $|\mathcal V_d|$ represents the cardinality of the discretized state set and $|\mathcal B_d|$ the cardinality of the discretized action set. The number of discretized state space is growing at least exponentially with the number of the sensors in order to keep the same grid size.
}

\subsection{Greedy Algorithm}
The optimal stochastic schedule is formidable to obtain in practice. Instead, we propose a suboptimal greedy schedule which minimizes the next-step trace of total expected error covariance. At each time $k$, the estimator broadcasts the priority queue $q(k)$ and the event-triggering parameters $\{\alpha_i(k)\}_{i\in\mathcal S}$ to all sensors which minimizes the one-step cost function:
\begin{align}
 &J_{greedy}(k)= \tr{\sum_{i\in\mathcal S} \E{P_i(k)}}
% \notag\\
% & = \p{\gamma_i(k)=1}\tr{\OP_i}+\p{\mu_i(k)=0}\tr{h_i(P_i(k-1))}\notag\\
% &~+\p{\mu_i(k)=1}\p{\eta_i(k)=0}\tr{t_i(P_i(k-1), \alpha_i(k))}.
\label{eq:jgreedy}.
\end{align}

Unlike searching the optimal priority queue for the optimal schedule via enumeration, we find a simple rule for the optimal priority queue for the greedy schedule if $r_i(k)=r_j(k),\forall i\neq j$. An example is that the system matrices of several identical vehicles have full rank, \ie, $\rank(A_i)=\rank(A_j)=n_i=n_j,\forall i,j\in\mathcal S$.
\begin{theorem}\label{thm:queue}
 If $r_i(k)=r_j(k)$ for any $i,j\in\mathcal S$, the queue $q(k)= \{s_{\varphi_1},\ldots,s_{\varphi_m},\ldots,s_{\varphi_n}\}$, where $\varphi_m\in\mathcal S$, is optimal for the greedy schedule if and only if for each $m< n$ the following inequality holds:
\begin{multline}
\tr{h_{\varphi_m}(P_{\varphi_m}(k-1))-\OP_{\varphi_m}}\\\geq \tr{h_{\varphi_{m+1}}(P_{\varphi_{m+1}}(k-1))-\OP_{\varphi_{m+1}}}.\label{eq:order_cond}
\end{multline}
\end{theorem}
\begin{proof}
 Note that the condition \eqref{eq:order_cond}, which applies to each adjacent pair in the queue, is a total order on the set $\mathcal S$. Therefore, if we prove ``only if'', then ``if'' is automatically proved. Next we will prove ``only if'' by contradiction.

Assume $q^+(k)$ to be optimal and there exists some adjacent pair in $q^+(k)$, without loss of generality, \ie, $(s_i,s_{i+1})$ with
\begin{align}
\tr{h_{i}(P_{i}(k-1))-\OP_{i}} < \tr{h_{i+1}(P_{i+1}(k-1))-\OP_{i+1}}\label{eq:19}
\end{align}
If we can show that by swapping $s_i$ and $s_{i+1}$ in $q^+(k)$ we can construct a better schedule $q^-(k)$ than $q^+(k)$, we will finish the proof by contradiction. We will discuss two cases: (I) $(s_{i},s_{i+1})$ is not the last pair in $q^+(k)$; (II) $(s_{i},s_{i+1})$ is the last pair.

%\textbf{Case I:}
Denote the quantities $\beta_i(k)$,$\beta_{i+1}(k)$,$\widehat{\alpha}_i(k)$,$\widehat{\alpha}_{i+1}(k)$ defined in \eqref{eq:alpha_def} and \eqref{eq:beta_def} under $q^+(k)~(q^-(k))$ as $\beta_i^+$,$\beta_{i+1}^+$,$\widehat{\alpha}_i^+$,$\widehat{\alpha}_{i+1}^+$$(\beta_i^-$,$\beta_{i+1}^-$,$\widehat{\alpha}_i^-$,$\widehat{\alpha}_{i+1}^-)$. We will omit the time index $k-1$ or $k$ if it is clear from the context.

\textbf{Case I}. Note that the next-step error covariance of sensors before $s_i$ in the queue are not affected by the swapping prodedure. In order not to disturb the next-step error covariance of sensors behind $s_3$, based on Theorem \ref{thm:pr}(iii) we let $\beta_i^+\beta_{i+1}^+=\beta_i^-\beta_{i+1}^-:=c$ where $c\in[0,1]$ is a constant such that the expected error covariances of all sensors except $s_i,s_{i+1}$ under $q^+,q^-$ are the same. From Theorem \ref{thm:optimalestimator} we have $J_{greedy}$ under $q^+,q^-$ are given as
\begin{align}
 &J_{greedy}^+\notag\\
 &= \mathrm{Tr}\left(\OP_i+h_{i+1}(P_{i+1})+(\widehat{\alpha}_i^+)^{\frac{\ell}{2}}\left[\widehat{\alpha}_i^+(h_i(P_i)-\OP_i)\right.\right.\notag\\
 &~~~~~\left.\left.+(c^{\frac{\ell+2}{\ell}}(\widehat{\alpha}_i^+)^{-\frac{\ell+2}{2}}-1)(h_{i+1}(P_{i+1})-\OP_{i+1})\right]\right)\label{eq:20}\\
 &J_{greedy}^- \notag\\
 &= \mathrm{Tr}\left(\OP_{i+1}+h_i(P_i)+(\widehat{\alpha}_{i+1}^-)^{\frac{\ell}{2}}\left[\widehat{\alpha}_{i+1}^-(h_{i+1}(P_{i+1})-\OP_{i+1})\right.\right.\notag\\
 &~~~~~~~\left.\left.+(c^{\frac{\ell+2}{\ell}}(\widehat{\alpha}_{i+1}^-)^{-\frac{\ell+2}{2}}-1)(h_i(P_i)-\OP_i)\right]\right).\label{eq:21}
\end{align}
Denote $\widehat{\alpha}_{i*}^+,\widehat{\alpha}_{(i+1)*}^-$ to be the minimizer of \eqref{eq:20} and \eqref{eq:21} and $J^+_*,J^-_*$ to be the corresponding costs. If we can show that
\begin{align*}
&\Delta J=J^+_*-J^-_*\\
& = \tr{h_i(P_i)-\OP_i}[\lambda(1+c^{\frac{\ell+2}{\ell}}(\widehat{\alpha}_{i*}^+)^{-1} - (\widehat{\alpha}_{i*}^+)^{\frac{\ell}{2}}-(\widehat{\alpha}_{(i+1)*}^-)^{\frac{\ell+2}{2}}) \\
&~~~- (1+c^{\frac{\ell+2}{\ell}}(\widehat{\alpha}_{(i+1)*}^-)^{-1} - (\widehat{\alpha}_{(i+1)*}^-)^{\frac{\ell}{2}} - (\widehat{\alpha}_{i*}^+)^{\frac{\ell+2}{2}})]\\
& \geq 0
\end{align*}
under \eqref{eq:19}, then the optimality of $q^+$ is violated. Next we will show that $\Delta J \leq 0$ for $\lambda<1$ and $\Delta J \geq 0$ for $\lambda>1$ holds.

First we need to know $\widehat{\alpha}_{i*}^+$ and $\widehat{\alpha}_{(i+1)*}^-$. By taking the first derivative of $J_{greedy}^+$ to be $0$, we have an equation
\begin{align}
 \Phi(x):=(\ell+2)x^{\frac{\ell+4}{2}}-\ell\lambda x^{\frac{\ell+2}{2}}-2\lambda c^{\frac{\ell+2}{\ell}}=0,~x>0\label{eq:22}
\end{align}
where $\lambda:=\frac{ \tr{h_{i+1}(P_{i+1})-\OP_{i+1}}}{\tr{h_{i}(P_{i})-\OP_{i}}}.$ By taking first and second derivatives of $\Phi(x)$, we know that $x=\frac{\lambda \ell}{\ell+4}$ is a global minimum on $(0,+\infty)$, $\Phi(x)$ is decreasing on $(0,\frac{\lambda \ell}{\ell+4})$ and $\Phi(x)$ is increasing on $(\frac{\lambda \ell}{\ell+4},+\infty)$. Since $\Phi(0)<0$, there is only one root for $\Phi(x)=0$, denoted by $x_*$. Therefore, we can conclude that $J_{greedy}^+$ is decreasing for $\widehat{\alpha}_i^+\in(0,x_*]$ and increasing for $\widehat{\alpha}_i^+\in(x_*,+\infty)$. From $\beta_i^+\beta_{i+1}^+=\beta_i^-\beta_{i+1}^-=c\in [0,1]$ and $\beta_i^+,\beta_{i+1}^+,\beta_i^-,\beta_{i+1}^-\in [0,1]$, we have $\widehat{\alpha}_{i*}^+ \in [c^{2/\ell},1]$. To sum up,
we have $\widehat{\alpha}_{i*}^+=\min(\max(x_*,c^{2/\ell}),1)$.

Similarly, $\widehat{\alpha}_{(i+1)*}^-=\min(\max(x_+,c^{2/\ell}),1)$ where $x_+$ is the positive solution to the following equation:
\begin{align}
 (\ell+2)x^{\frac{\ell+4}{2}}-\ell\tilde{\lambda} x^{\frac{\ell+2}{2}}-2\tilde{\lambda} c^{\frac{\ell+2}{\ell}}=0,~x>0\label{eq:23}
\end{align}
where $\tilde{\lambda} = \frac{1}{\lambda}$.

From \eqref{eq:22} we have
\begin{align}
\frac{\dd x_*}{\dd \lambda } = \frac{2c^{\frac{\ell+2}{\ell}}+\ell(x_*)^{\frac{\ell+2}{2}}}{(\ell+2)(x_*)^{-1}\left((x_*)^{\frac{\ell+4}{2}}+c^{\frac{\ell+2}{\ell}}\lambda\right)}\geq 0.
\label{eq:30}
\end{align}

Since $x_*$ is a nondecreasing function of $\lambda$, there is $\lambda^*=(\ell+2)/(\ell+2 c^{\frac{\ell+2}{\ell}})>1$ such that $\widehat{\alpha}_{i*}^+=1$ for all $\lambda>\lambda^*$. When $\lambda>\lambda^*$, \ie, $\widehat{\alpha}_{i*}^+=1$, based on \eqref{eq:20}, \eqref{eq:21}, \eqref{eq:23} we can derive that
\begin{align*}
\Delta J &= \tr{h_{2}(P_{2})-\OP_{2}}\frac{\ell+2}{2} (\widehat{\alpha}_{(i+1)*}^-)^{\frac{\ell}{2}}(\lambda \widehat{\alpha}_{(i+1)*}^- -1)^2\\
&\geq 0.
\end{align*}
Similarly, when $\lambda<1/\lambda^*$, $\Delta J \leq 0$.

Now we consider the case of $1/\lambda^*\leq\lambda\leq \lambda^*$. Denote $g(\lambda) = \frac{\Delta J}{\tr{h_i(P_i)-\OP_i}}$ and its first and second derivatives are $g'(\lambda)$ and $g''(\lambda)$.
Then we have
\begin{align}
g'(\lambda) &= 1 + c^{\frac{\ell+2}{\ell}}(\widehat{\alpha}_{i*}^+)^{-1}-(\widehat{\alpha}_{i*}^+)^{\frac{\ell}{2}}-(\widehat{\alpha}_{(i+1)*}^-)^{\frac{\ell+2}{2}} \label{eq:26}\\
g''(\lambda) &= \frac{\ell+2}{2\lambda}\left(\frac{1}{\lambda}(\widehat{\alpha}_{(i+1)*}^-)^{\frac{\ell}{2}}\frac{\dd \widehat{\alpha}_{(i+1)*}^-}{\dd (1/\lambda)}-(\widehat{\alpha}_{i*}^+)^{\frac{\ell}{2}}\frac{\dd\widehat{\alpha}_{i*}^+}{\dd \lambda}\right)\label{eq:29}
\end{align}

If we can prove that $g(\lambda)$ is convex for $1/\lambda^*\leq\lambda<1$, \ie, $g''(\lambda) \geq 0$. Since $g(1/\lambda^*)<0$ and $g(1)=0$, we have that $g(\lambda)<0$ for $1/\lambda^*\leq\lambda<1$. Moreover, if we can prove that $g(1/\lambda) = -g(\lambda)/\lambda$, we have $g(\lambda)>0$ for $1 < \lambda< \lambda^*$. Namely, we can show that $\Delta J=0$ has a unique root at $\lambda=1$. First we prove the convexity.
From \eqref{eq:22} we have
\begin{align*}
\frac{\dd \widehat{\alpha}_{i*}^+}{\dd \lambda } &= \frac{2c^{\frac{\ell+2}{\ell}}+\ell(\widehat{\alpha}_{i*}^+)^{\frac{\ell+2}{2}}}{(\ell+2)(\widehat{\alpha}_{i*}^+)^{-1}\left((\widehat{\alpha}_{i*}^+)^{\frac{\ell+4}{2}}+c^{\frac{\ell+2}{\ell}}\lambda\right)}\\
\frac{\dd \widehat{\alpha}_{(i+1)*}^-}{\dd (1/\lambda) } &= \frac{2c^{\frac{\ell+2}{\ell}}+\ell(\widehat{\alpha}_{(i+1)*}^-)^{\frac{\ell+2}{2}}}{(\ell+2)(\widehat{\alpha}_{(i+1)*}^-)^{-1}\left((\widehat{\alpha}_{(i+1)*}^-)^{\frac{\ell+4}{2}}+c^{\frac{\ell+2}{\ell}}/\lambda\right)}
\end{align*}
Then from \eqref{eq:29} we have
\begin{align}
g''(\lambda) &= \kappa [ 2c^{\frac{\ell+2}{\ell}}(\widehat{\alpha}_{i*}^+\widehat{\alpha}_{(i+1)*}^-)^{\frac{\ell+2}{2}}(\widehat{\alpha}_{i*}^+-\lambda\widehat{\alpha}_{(i+1)*}^-) \label{eq:30a}\\
&~~+2c^{\frac{2\ell+4}{\ell}}(\lambda (\widehat{\alpha}_{(i+1)*}^-)^{\frac{\ell+2}{2}} -(\widehat{\alpha}_{i*}^+)^{\frac{\ell+2}{2}})\label{eq:30b}\\
&~~+\ell (\widehat{\alpha}_{i*}^+\widehat{\alpha}_{(i+1)*}^-)^{\frac{\ell+4}{2}}( (\widehat{\alpha}_{(i+1)*}^-)^{\frac{\ell}{2}} -\lambda(\widehat{\alpha}_{i*}^+)^{\frac{\ell}{2}})\label{eq:30c}\\
&~~+ c^{\frac{\ell+2}{\ell}}\ell (\lambda (\widehat{\alpha}_{(i+1)*}^-)^{\ell+2} -(\widehat{\alpha}_{i*}^+)^{\ell+2})]\label{eq:30d},
\end{align}
where $\kappa$ is some positive term. Now we will show that the terms in \eqref{eq:30a}-\eqref{eq:30d} are all positive.

We first prove $(\widehat{\alpha}_{i*}^+-\lambda\widehat{\alpha}_{(i+1)*}^-)$ in \eqref{eq:30a} is positive. Since there is only one root for \eqref{eq:22}, the fact that
$\widehat{\alpha}_{i*}^+-\lambda\widehat{\alpha}_{(i+1)*}^->0$ is equivalent to that $\Phi(\lambda\widehat{\alpha}_{(i+1)*}^-)<0$. We have that
\begin{align*}
&\Phi(\lambda\widehat{\alpha}_{(i+1)*}^-)\\
& = (\ell+2)(\lambda\widehat{\alpha}_{(i+1)*}^-)^{\frac{\ell+4}{2}}-\ell\lambda (\lambda\widehat{\alpha}_{(i+1)*}^-)^{\frac{\ell+2}{2}}-2\lambda c^{\frac{\ell+2}{\ell}}\\
& = \lambda (\widehat{\alpha}_{(i+1)*}^-)^{\frac{\ell+2}{2}}(\lambda^{\frac{\ell+2}{2}}-\lambda)(\ell+2)\\
&~~~~~~\times\left[ \widehat{\alpha}_{(i+1)*}^- -\frac{(\lambda^{\frac{\ell+2}{2}}-1)\ell}{(\lambda^{\frac{\ell+2}{2}}-\lambda)(\ell+2)}\right].
\end{align*}
We can easily show that $\Lambda(\lambda)=\frac{(\lambda^{\frac{\ell+2}{2}}-1)\ell}{(\lambda^{\frac{\ell+2}{2}}-\lambda)(\ell+2)}$ is a decreasing function of $\lambda$ for $\lambda>0$. Since $\lambda<1$, we have the minimum of $\Lambda(\lambda)$ at $\lambda=1$. By using L'Hospital's rule, we have $\lim_{\lambda\rightarrow 1^-} = 1$. Since $\widehat{\alpha}_{(i+1)*}^-\leq 1$, we have $\Phi(\lambda\widehat{\alpha}_{(i+1)*}^-)<0$ and thus $(\widehat{\alpha}_{i*}^+-\lambda\widehat{\alpha}_{(i+1)*}^-)$ in \eqref{eq:30a} is positive.

Now we prove $(\lambda (\widehat{\alpha}_{(i+1)*}^-)^{\frac{\ell+2}{2}} -(\widehat{\alpha}_{i*}^+)^{\frac{\ell+2}{2}})$ in \eqref{eq:30b} is positive. From \eqref{eq:30} and $\lambda<1$, we have that $\widehat{\alpha}_{(i+1)*}^- > \widehat{\alpha}_{i*}^+$. From \eqref{eq:22} and \eqref{eq:23}, we have
\begin{align*}
&(\ell+2)(\lambda^2 (\an)^{\frac{\ell+4}{2}} - (\ap)^{\frac{\ell+4}{2}})\\
&~~~~~~~~~~=\lambda \ell( (\an)^{\frac{\ell+2}{2}}- (\ap)^{\frac{\ell+2}{2}}) > 0.
\end{align*}
Due to $\lambda^2 (\an)^{\frac{\ell+4}{2}} > (\ap)^{\frac{\ell+4}{2}} >0$ and $\ap>\lambda\an$, we have $\lambda (\an)^{\frac{\ell+2}{2}} > (\ap)^{\frac{\ell+2}{2}}$.

It is easy to see that $(\an)^{\frac{\ell}{2}}>\lambda (\ap)^{\frac{\ell}{2}}$. Due to $\lambda^2 (\an)^{\frac{\ell+4}{2}} > (\ap)^{\frac{\ell+4}{2}} >0$ and $\an>\ap$ and $\lambda<1$, we have $\lambda (\an)^{\ell+2} > (\ap)^{\ell+2}$.

Thus we can show that $g''(\lambda)\geq 0$ and $g(\lambda)$ is convex for $1/\lambda^*\leq\lambda<1$. From \eqref{eq:20}-\eqref{eq:23}, it is easy to show that $g(1/\lambda) = -g(\lambda)/\lambda$. Hence we complete the proof of the fact that $\Delta J=0$ has a unique root at $\lambda=1$.

\textbf{Case II}. The last pair $(s_{n-1}, s_n)$ needs special attention because the last sensor does not check its own data importance. Since the transmission of the last sensor only depends on its preceding sensor, we have $\eqref{eq:20}\eqref{eq:21}$ rewritten into
\begin{align}
 & J_{greedy}^+= \mathrm{Tr}\left( \OP_{n-1}+h_n(P_n)+ (\widehat{\alpha}_{n-1}^+)^{\frac{\ell}{2}} \right.\notag\\
 &~~~\left.\times[\widehat{\alpha}_{n-1}^+(h_{n-1}(P_{n-1})-\OP_{n-1})-(h_{n}(P_{n})-\OP_{n})] \right)\label{eq:31}\\
 & J_{greedy}^-= \mathrm{Tr}\left(\OP_{n}+h_{n-1}(P_{n-1})+ (\widehat{\alpha}_{n}^-)^{\frac{\ell}{2}} \right. \notag\\
 &~~~\left.\times[\widehat{\alpha}_{n}^-(h_{n}(P_{n})-\OP_{n})-(h_{n-1}(P_{n-1})-\OP_{n-1})]\right) \label{eq:32}
\end{align}
In this case, we can see that $c=0$ compared with $\eqref{eq:20}\eqref{eq:21}$.
We can find the minimizer of $J_{greedy}^+$ by taking the first derivative of $J_{greedy}^+$ and we have the minimizer $\apn$ as follows:
\begin{align*}
\begin{cases}
     \apn = 1 & \text{if } \frac{\ell}{\ell+2}\lambda>1, \\
     \apn = \frac{\ell}{\ell+2}\lambda & \text{if } 0 < \frac{\ell}{\ell+2}\lambda\leq 1,
\end{cases}
\end{align*}
where $\lambda = \frac{ \tr{h_{n}(P_{n})-\OP_{n}}}{\tr{h_{n-1}(P_{n-1})-\OP_{n-1}}}$
Similarly, we have the minimizer $\ann$ for $J_{greedy}^-$ as follows:
\begin{align*}
\begin{cases}
     \ann = 1 & \text{if } \frac{\ell}{\ell+2}\frac{1}{\lambda}>1, \\
     \ann = \frac{\ell}{\ell+2}\frac{1}{\lambda} & \text{if } 0 < \frac{\ell}{\ell+2}\frac{1}{\lambda}\leq 1.
\end{cases}
\end{align*}

Since $\apn,\ann$ depend on the value of $\lambda$, we discuss $\Delta J$ for different ranges of $\lambda$ like we did in Case I.

When $\frac{\ell}{\ell+2} \leq \lambda \leq \frac{\ell+2}{\ell}$, we have
\begin{align*}
    \Delta J &= \tr{h_{n-1}(P_{n-1}) - \OP_{n-1}}\\
    &~~~\times \left(\lambda-1-\frac{2}{\ell}\left(\frac{\ell\lambda}{\ell+2}\right)^{\frac{\ell+2}{2}}+\frac{2}{\ell+2}\left(\frac{2}{\ell+2}\frac{1}{\lambda}\right)^{\frac{\ell}{2}}\right).
\end{align*}
With some calculations, we know that there is only one solution to $\dd^2 \Delta J/\dd \lambda^2 = 0.$ Moreover, we have $\dd \Delta J/ \dd \lambda < 0 $ at $\lambda=\frac{\ell}{\ell+2}$ and $\frac{\ell+2}{\ell}$. Then we can conclude that there are two roots or no root for $\dd \Delta J/ \dd \lambda = 0$. If there are two roots $r_1$ and $r_2$, it must be true that $\dd \Delta J/ \dd \lambda<0$ for $\lambda<r_1$ or $\lambda>r_2$, and $\dd \Delta J/ \dd \lambda>0$ for $r_1<\lambda<r_2$.

Since $\Delta J<0$ at $\lambda = \frac{\ell}{\ell+2}$ and $\Delta J>0$ at $\frac{\ell+2}{\ell}$, there must be odd number of roots for $\Delta J=0$. If there is no root for $\frac{\dd \Delta J}{\dd \lambda}=0$, then $\Delta J$ is decreasing function which is impossible. So there are two roots for $\frac{\dd \Delta J}{\dd \lambda}=0$. From the Rolle's theorem, there is only one root for $\Delta J=0$. By inspection we have one root $\lambda=1$ and it is unique based on our previous reasoning. Therefore, we have that
\[
\begin{cases}
    \Delta J>0  & \text{if } \frac{\ell+2}{\ell}\geq \lambda > 1, \\
    \Delta J\leq 0  & \text{if } 1\geq \lambda >\frac{\ell}{\ell+2}.
\end{cases}
\]
So if $\Delta J>0$, we choose $2,1$. Otherwise, we choose $1,2$.

When $\lambda\geq \frac{\ell+2}{\ell}$, we have
\begin{align*}
  \Delta J &=  \tr{h_{n-1}(P_{n-1}) - \OP_{n-1}} \left[\frac{2}{\ell+2}\left(\frac{\ell}{\ell+2}\frac{1}{\lambda}\right)^{\ell/2}\right]\\
 &> 0.
\end{align*}
When $\lambda \leq \frac{\ell}{\ell+2}$, we have 
\begin{align*}
  \Delta J &=  \tr{h_{n-1}(P_{n-1}) - \OP_{n-1}} \left[\frac{2}{\ell+2}\left(-\frac{2}{\ell}\lambda\right)^{\frac{\ell+2}{2}}\right]\\
 &> 0.
\end{align*}

Thus in summary, we can conclude that $\Delta J>0$ if $\lambda>1$ and $\Delta J\leq 0 $ if $\lambda\leq 1$.

Hence the conclusion that $\Delta J<0$ for $\lambda<1$ and $\Delta J>0$ for $\lambda>1$ holds for both Case I and Case II. Since the above proof applies to any adjacent pair in any feasible schedule (no matter it is optimal or not), we can always swap the pair like $s_i,s_{i+1}$ to improve the performance if a condition like \eqref{eq:19} holds. Thus we prove the ``only if'' part by contradiction which completes the proof.
\end{proof}

%Next we will show that $\lambda=1$ is the unique solution to $\Delta J=0$. Assume $\lambda^\dagger>1$ is the largest root of $\Delta J=0$ and correspondingly $1/\lambda^\dagger$ is the least root of $\Delta J=0$. Then we must have that $d\Delta J/d\lambda>0$ at $\lambda=\lambda^\dagger$ and $1/\lambda^\dagger$. From \eqref{eq:20}, \eqref{eq:21}, \eqref{eq:22} and \eqref{eq:23}, we have
%\begin{align*}
% \frac{d\Delta J}{d\lambda}\Big|_{\lambda=\lambda^\dagger} = -\frac{d\Delta J}{d\lambda}\Big|_{\lambda=1/\lambda^\dagger}
%\end{align*}
%which cannot be positive simultaneously, and thus there must exist only one root of $\lambda=1$. To sum up, we can conclude that $\Delta J>0$ if $\lambda>1$ and $\Delta J<0$ if $\lambda<1$.

%\begin{remark}
%If $r_i(k)\neq r_j(k)$, the thresholding rule like \eqref{eq:order_cond} exists but it also depends on $r_i(k), r_j(k)$ of the two adjacent sensors in a complicated manner, \ie, one cannot put all the quantities about $s_i$ on the LHS and all those about $s_j$ on the RHS like \eqref{eq:order_cond}. In that case, the condition is no longer a total order and thus only a necessary condition for the optimal queue can be derived.
%\end{remark}
%\begin{remark}
%If $r_i(k)\neq r_j(k)$, the condition is no longer a total order and thus only a necessary condition for the optimal queue can be derived.
%\end{remark}

\bb{After determining the optimal $q(k)$, we can optimize $J_{greedy}(k)$ in \eqref{eq:jgreedy} in terms of $\{\alpha_i(k)\}_{i\in\mathcal S}$. Based on Theorem \ref{thm:optimalestimator} and \ref{thm:pr}, we can explicitly obtain a cost function which turns out to be a signomial of $\{\alpha_i(k)\}_{i\in\mathcal S}$. The signomial programming (SP) problem is widely studied in communication society \cite{chiang2005geometric}. Though it is not convex, there are many efficient algorithms for local optimum or global optimum (see, e.g., \cite{chiang2005geometric,maranas1997global}). The computational efforts required for different signomial programming problems are illustrated with examples in \cite{maranas1997global}.
}

%In Algorithm \ref{algorithm:1}, we outline the greedy method for suboptimal parameter design.
%
%\begin{algorithm}
%\caption{Greedy Algorithm}
%\label{algorithm:1}
%\begin{algorithmic}[1]
%\State {\bf Initialization.}
%\State {\bf Repeat}
%\begin{enumerate}
%\item If $r_i(k)=r_j(k)$ for any $i,j\in\mathcal S$, determine the optimal priority queue based on Theorem \ref{thm:queue}. Otherwise, repeat 2),3) by enumerating all possible $q(k)$'s.
%\item Compute the optimal $\{\alpha_i(k)\}_{i\in\mathcal S}$.
%\item If the optimal priority queue has been obtained based on Theorem \ref{thm:queue}, jump to 5). Otherwise, jump to 2) until all queues are enumerated.
%\item Pick up the optimal $\{\alpha_i(k)\}_{i\in\mathcal S}$ and $q(k)$ such that $J_{greedy}(k)$ is minimized.
%\item Broadcast the parameters and receive sensory data at time $k$. Update $P_i(k),~\forall i\in\mathcal S$.
%\end{enumerate}
%\end{algorithmic}
%\end{algorithm}

\subsection{Optimality Gap}
Notice that the performance gap between the suboptimal schedules and the optimal $\theta_+\in\Theta_+$ is upperbounded by the gap between the suboptimal performance and a lower bound for the optimal performance. Next we discuss the upper bound of such performance gap to quantify the performance of any suboptimal stochastic schedule.

First we construct an artificial schedule whose performance is better than the optimal $\theta_+\in\Theta_+$. We relax the original constraint $\sum_i \gamma_i(k)=1$ by requiring the sum of the rates of all sensors to be less than $1$. In other words, we allow that channel can be used by multiple sensors but their sum communication rate must be less than $1$. Denote $\Theta_\ddagger$ as the set of all schedules in \eqref{eq:event} with forcing $\mu_i(k)$ to be $1$ for $i\leq n-1$ and all $k$. In other words, any absence of $\hat{x}_{i,{\rm local}}(k)$ implies \eqref{eq:eta} for $i\leq n-1$. Moreover, we also implement the event-triggering mechanism for $s_n$. By solving the following optimization problem:
\begin{problem} \label{problem:relax-problem-1}
\begin{align}
& \underset{\theta\in\Theta_\ddagger}{\textit{minimize}}
& & J(\theta), & \textit{subject to}& &\sum_{i\in\mathcal S}\E{\gamma_i(k)} \leq 1,\label{eq:12}
\end{align}
\end{problem}
it is not hard to see the solution to Problem \ref{problem:relax-problem-1} is better than the optimal $\theta_+\in\Theta_+$ since any absence of each sensor now conveys some information to the estimator. While  for $\theta_+\in\Theta_+$ the information is lost for the sensors whose priority is behind the designated one. Also note that the queue is useless for $\theta\in\Theta_\ddagger$.

Now $\alpha_i(k)$ solely depends on $\tau_i(k)$ and so does $\beta_i(k)$. From \eqref{eq:3} we know that $J_i(\theta)$ is determined by the underlying stochastic process $\{\tau_i(k)\}_{k=1}^\infty$. With a little abuse of notations, we use $\beta_i^j$ to denote the probability of $\p{\eta_i(k)=0}$ for sensor $i$ with $\tau_i(k)=j$ and $\alpha_i^j$ accordingly for $i\leq n$.

We can quantify the performance of a suboptimal schedule by using the following result.
\begin{proposition}\label{prop:1}
Suppose $\theta_*\in\Theta_+$ to be an optimal solution of minimizing \eqref{eq:cost-function-total}. For any schedule $\theta_\bot\in\Theta_+$ (including any time-based schedules), define the optimality gap to be $\Delta(\theta_\bot):=J(\theta_\bot)-J(\theta_*)$. The following inequality holds:
\begin{align*}
\Delta(\theta_\bot) \leq J(\theta_\#)-J(\theta_*).
\end{align*}
The schedule $\theta_\#$ is characterized by $\{\beta_i^\ell\}_{i\in\mathcal S,\ell\in\mathbb Z_+}$ which is the solution to the following problem:
% \footnote{Note that we need to solve the problem over countable $\beta_i^\ell$'s. We can set an upper bound on $l$ for any $i$ in practice. Intuitively, $\beta_i^\ell$ with a large $l$ has negligible effects on the cost function since $\pi_i^\ell$ depends on the product of all $\beta_i^j\in[0,1]$, $j\leq l$.}
\begin{align*}
&\underset{\pi_i^0,\beta_i^\ell\in\mathbb R}{\textit{minimize}}
& & \sum_{i\in\mathcal S}\pi_i^0\tr{P_i^0}+\sum_{j=0}^{\infty}\pi_i^0(\prod_{\ell=0}^j\beta_i^\ell)\tr{P_i^{j+1}},\\
& \textit{subject to} & &\pi_i^0+\pi_i^0\sum_{j=0}^{\infty}\prod_{\ell=0}^j\beta_i^\ell\geq 1,~i\in\mathcal S, \\
& & &\sum_{i\in\mathcal S} \pi_i^0\leq 1,
\end{align*}
where
\begin{align}
P_i^j &:= \prod_{u=0}^{j-1}(\beta_i^u)^{2/r_i(u)}h_i^j(\OP_i)+\sum_{\ell=0}^{j-1}\left(1-(\beta_i^{j-1-\ell})^{2/r_i(j-1-\ell)}\right) \notag\\
&\times \left(\prod_{u=0}^{\ell-1}(\beta_i^{j-1-u})^{2/r_i(j-1-\ell)}\right)h_i^\ell(\OP_i),~~\forall i\in\mathcal S. \label{eq:pij}
\end{align}
\end{proposition}
\begin{proof}In Problem \ref{problem:relax-problem-1}, the transmission of each sensor is not affected by others and we can first consider such a subproblem for one sensor:
\begin{align}
& \underset{\theta\in\Theta_\ddagger}{\textit{minimize}}
& &  J_i(\theta), & \textit{subject to}& &  \E{\gamma_i(k)} \leq \varrho_i. \label{eq:13}
\end{align}

The stochastic process $\{\tau_i(k)\}_{k=1}^\infty$ is a Markov chain with countable state space $\mathcal M:=\{0,1,2,\ldots\}$. The transition matrix $T_i$ is given by
\begin{align*}
T_i(j_2,j_1)=\left\{
               \begin{array}{ll}
                 1-\beta_i^{j_1}, & \hbox{if~} j_2=0\\
                 \beta_i^{j_1}, & \hbox{if~} j_2=j_1+1\\
                 0, & \hbox{otherwise}
               \end{array}
             \right.,
\end{align*}
where $T_i(j_2,j_1):=\p{\tau_i(k)=j_2|\tau_i(k)=j_1}$. Though we can design $\beta_i^j$ freely, we have to choose the set of $\{\beta_i^j\}_{j=0}^{\infty}$ to guarantee that the state $\tau_i(k)=0$ is \textit{recurrent}. Otherwise, the Markov chain $\{\tau_i(k)\}$ is transient and $J_i(\theta)$ will go unbounded due to Lemma \ref{lemma:supporting-lemma}(iii) and the fact that $\lim_{j\rightarrow \infty}\tr{h_i^j(\OP_i)}\rightarrow \infty$. Therefore, the stationary distribution $\pi_i=[\pi_i^0,\pi_i^1,\pi_i^2,\ldots]$ must exist \cite{meyn2012markov}. Denote the corresponding estimation error covariance for $\tau_i(k)$ as $P_i^{\tau_i(k)}$. From Theorem \ref{thm:optimalestimator} and \ref{thm:pr}(i) we have that $P_i^0 = \OP_i,P_i^{j} = t(P_i^{j-1},\alpha_i^{j-1})$ and thus $P_i^j$ in \eqref{eq:pij}. Now the cost function becomes $J_i(\theta)=\varrho_i\tr{P_i^0}+\sum_{j=0}^{\ell_{max}-1}\varrho_i(\prod_{\ell=0}^j\beta_i^\ell)\tr{P_i^{j+1}}$.

Note that the implicit equality constraint for Problem \ref{problem:relax-problem-1} is $\sum_j \pi_i^j =1$. Since the cost function $J_i(\theta)$ is a decreasing function of $\pi_i^0$, we can transform it into the inequality constraint and the optimal solution must be reached only when the equality holds. Therefore, by solving the optimazition problem in the proposition we can obtain a lower bound for the optimal cost and thus the optimality gap bound.
\end{proof}
Note that the optimization problem is also a signomial programming problem whose global optimum can be numerically solved \cite{chiang2005geometric,maranas1997global}.

\begin{remark}\label{remark:queue}
For a small network, we can enumerate the optimal or suboptimal schedules for different orders and choose the best order. However, it is formidable to do this for a large network. One heuristic technique is to assign a time-invariant priority queue in parameter design according to the optimal $\pi_i^0$ in Proposition \ref{prop:1}, \ie, letting the queue to be $\{s_{\phi_1},\ldots,s_{\phi_\ell},\ldots, s_{\phi_n}\}$,~where $\phi_\ell\in\mathcal S$, with $\pi_{\phi_\ell}^0\leq\pi_{\phi_{\ell+1}}^0$ since a larger $\pi_{\phi_\ell}^0$ usually implies that error covariance of $s_{\phi_\ell}$ may grow faster than others.
\end{remark}

%Under the time-invariant priority queue, note that the computational burden of the suboptimal schedule and the lower bound is scalable to the size of the network while the MDP method is subject to the curse of dimensionality.

\section{Example}\label{section:simulation}
\bb{In this section, we use a simple example to show the superiority of the proposed event-based schedule and compare event-based schedules with different parameter designs. To compare with the optimal time-based schedule in \cite{shi2012scheduling} which is only applicable to the two-sensor case, we also conduct simulations over two processes with the parameters: $A_1=\begin{bmatrix}2 & 1\\ 0& 1\end{bmatrix},~C_1 = \begin{bmatrix}1\\ 2\end{bmatrix}\T,~Q_1=\begin{bmatrix}1 & 0\\ 0& 1\end{bmatrix},~R_1=1$
 and $A_2=\begin{bmatrix}1.1 & 1\\ 0& 1\end{bmatrix},~C_2 = \begin{bmatrix}1\\ 1\end{bmatrix}\T,~Q_2=\begin{bmatrix}3 & 0\\ 0& 3\end{bmatrix},~R_2=1$.}

We compare the performance among the following schedules:
\begin{itemize}
\item $\theta_{\textit{offline}}$: the optimal time-based schedule\cite{shi2012scheduling} is $\{s_2,s_1,s_1,s_2,s_1,s_1,\ldots\}$ with the period $3$.
\item $\theta_{\textit{greedy}}$: the suboptimal schedule via the Greedy algorithm.
\item $\theta_{\textit{MDP}}$: the approximate optimal event-based schedule via MDP approach. By discretizing the continuous state and action space and rebuilding the transition law, we approximate the original MDP problem by a finite state discrete MDP problem and solve it by value iteration. The number of state grids and action grids is $1000$ and $100$, respectively.
\end{itemize}

\begin{table}
\centering
 \begin{tabular}{ccccccc}
 \toprule
   Schedules and LB &$\theta_{\textit{offline}}$ &$\theta_{\textit{greedy}}$ &$\theta_{\textit{MDP}}$ &LB\\
   $J(\cdot)$&$ 92.64$ & $52.05$ &$55.23$ & $48.21$\\ \bottomrule
  \end{tabular}
  \caption{Performance comparison among different schedules. The lower bound of the optimal $J$ is also listed as LB.}\label{table:1}
  %\vspace{-5mm}
\end{table}

% 1.8062

We also compute the lower bound for the optimal event-based schedule according to Proposition \ref{prop:1}.

\bb{We summarize the results in Table \ref{table:1}. For performance of each schedule, we repeatedly run $500$ experiments and take their average. The estimation performance of $\theta_{\textit{offline}}$ is the worst among all. The two event-based schedules reduce the cost function $J(\theta_{\textit{offline}})$ by $43.8\%$ and $40.4\%$, respectively. The approximate solution to the MDP does not outperform the solutions given by the greedy algorithm in this case due to the discretization approximation. Moreover, the gap between $J(\theta_{\textit{greedy}})$ and LB is small, implying that the suboptimal schedule via the greedy algorithm is a good choice with low computational complexity compared with the MDP-based optimal schedule. The computational time for $\theta_{\textit{greedy}}$ at each time and the lower bound are $3.3$s and $25.1$s, respectively\footnote{All simulations are conducted on MacBook with a processor of $1.3$GHz Intel Core i5 and a memory of $4$GB DDR3.}.}

%\setlength{\figureheight}{2.3cm}
%\setlength{\figurewidth}{6.5cm}
%\begin{figure}
%  \centering
%  %\inputtikz{realization}
%  % Requires \usepackage{graphicx}
%  \includegraphics[width= 3.7 in, height = 2.5in]{fig.png}
%  \caption{Realization of the event-based schedule $\theta_{\textit{greedy}}$ (upper figure) and realization of $\abs{\epsilon_1(k)}$ (lower figure).} \label{fig:2}
%  %\vspace{-5mm}
%\end{figure}

\bb{We also plot the realization of $\theta_{greedy}$ to illustrate the superiority of the event-based mechanism over the time-based one. At time $k=11$, Sensor $2$ has the first priority of using the channel in Fig. \ref{fig:2}(b). This reflects the validity of Theorem \ref{thm:queue} from the fact $\tr{h_2(P_2(k))}-\OP_2> \tr{h_1(P_1(k))}-\OP_1$ seen from Fig. \ref{fig:2}(a). In Fig. \ref{fig:2}(c), however, Sensor $1$ is scheduled to transmit its data because the event-triggering mechanism of Sensor $2$ classifies its data as unimportant and leaves the channel access to Sensor $1$. The design of queue and event-triggering thus allocates the channel access more efficiently than the offline schedule does.}

\setlength\figureheight{2.6cm}
\setlength\figurewidth{8cm}
\begin{figure}
  \centering
  \input{realization.tex}
  \caption{Realization of the event-based schedule $\theta_{\textit{greedy}}$: a) $\tr{P_i(k)}$ of each sensor, b) y-axis denotes the sensor index with the higher priority, c) y-axis denotes the on-duty sensor index.} \label{fig:2}
  %\vspace{-5mm}
\end{figure}
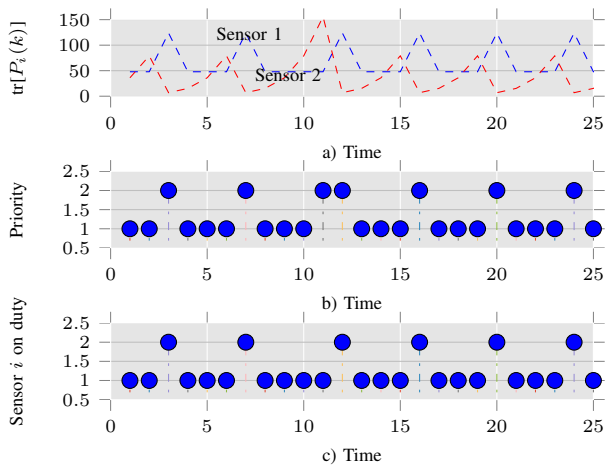

\section{Conclusion}\label{section:conclusion}
We have studied a stochastic sensor scheduling framework for sensor networks monitoring different processes. The sensors make transmission decisions based on both channel accessibility and self-triggering events depending on the realtime innovations. The proposed schedule which dynamically allocates the limited communication bandwidth to sensors is shown to be more efficient than the time-based schedules in terms of average estimation performance. We have also discussed the optimal event-based schedule design through an MDP approach. Since solving an MDP problem in Borel spaces is generally computationally difficult, we have also proposed a suboptimal schedule via the greedy algorithm and analyzed the optimality gap. Future works include the design of event-based scheduling over different communication topologies and conditions, \ie, different graphs or imperfect channels.

\bibliographystyle{IEEETran}
\bibliography{csma_references}

\end{document}

%% file: realization.tex
% This file was created by matplotlib2tikz v0.5.15.
\begin{tikzpicture}

\definecolor{color5}{rgb}{1,0.709803921568627,0.72156862745098}
\definecolor{color4}{rgb}{0.556862745098039,0.729411764705882,0.258823529411765}
\definecolor{color0}{rgb}{0.886274509803922,0.290196078431373,0.2}
\definecolor{color3}{rgb}{0.984313725490196,0.756862745098039,0.368627450980392}
\definecolor{color2}{rgb}{0.596078431372549,0.556862745098039,0.835294117647059}
\definecolor{color1}{rgb}{0.203921568627451,0.541176470588235,0.741176470588235}

\begin{groupplot}[group style={group size=1 by 3}]
\nextgroupplot[
xlabel={a) Time},
ylabel={tr[$P_i(k)$]},
xmin=0, xmax=25,
ymin=0, ymax=150,
width=\figurewidth,
height=\figureheight,
tick align=outside,
xmajorgrids,
x grid style={white},
ymajorgrids,
y grid style={white},
axis line style={white},
axis background/.style={fill=white!89.803921568627459!black}
]
\addplot [blue, dashed]
table {%
1 47.8929495560746
2 47.8929495560746
3 123.866124381858
4 47.8929495560746
5 47.8929495560746
6 47.8929495560746
7 123.866124381858
8 47.8929495560746
9 47.8929495560746
10 47.8929495560746
11 47.8929495560746
12 123.866124381858
13 47.8929495560746
14 47.8929495560746
15 47.8929495560746
16 123.866124381858
17 47.8929495560746
18 47.8929495560746
19 47.8929495560746
20 123.866124381858
21 47.8929495560746
22 47.8929495560746
23 47.8929495560746
24 123.866124381858
25 47.8929495560746
26 47.8929495560746
27 47.8929495560746
28 123.866124381858
29 47.8929495560746
30 47.8929495560746
31 47.8929495560746
32 47.8929495560746
33 123.866124381858
34 47.8929495560746
35 47.8929495560746
36 47.8929495560746
37 123.866124381858
38 47.8929495560746
39 47.8929495560746
40 47.8929495560746
41 123.866124381858
42 47.8929495560746
43 47.8929495560746
44 47.8929495560746
45 123.866124381858
46 47.8929495560746
47 47.8929495560746
48 47.8929495560746
49 123.866124381858
50 47.8929495560746
};
\addplot [red, dashed]
table {%
1 36.0572918170819
2 79.0923029925647
3 6.87916820220319
4 15.3190073603705
5 36.0572918170819
6 79.0923029925647
7 6.87916820220319
8 15.3190073603705
9 36.0572918170819
10 79.0923029925647
11 155.759678656108
12 6.87916820220319
13 15.3190073603705
14 36.0572918170819
15 79.0923029925647
16 6.87916820220319
17 15.3190073603705
18 36.0572918170819
19 79.0923029925647
20 6.87916820220319
21 15.3190073603705
22 36.0572918170819
23 79.0923029925647
24 6.87916820220319
25 15.3190073603705
26 36.0572918170819
27 79.0923029925647
28 6.87916820220319
29 15.3190073603705
30 36.0572918170819
31 79.0923029925647
32 155.759678656108
33 6.87916820220319
34 15.3190073603705
35 36.0572918170819
36 79.0923029925647
37 6.87916820220319
38 15.3190073603705
39 36.0572918170819
40 79.0923029925647
41 6.87916820220319
42 15.3190073603705
43 36.0572918170819
44 79.0923029925647
45 6.87916820220319
46 15.3190073603705
47 36.0572918170819
48 79.0923029925647
49 6.87916820220319
50 15.3190073603705
};
\node at (axis cs:5,111)[
  scale=1,
  anchor=base west,
  text=black,
  rotate=0.0
]{ Sensor 1};
\node at (axis cs:7,30)[
  scale=1,
  anchor=base west,
  text=black,
  rotate=0.0
]{ Sensor 2};
\nextgroupplot[
xlabel={b) Time},
ylabel={Priority},
xmin=0, xmax=25,
ymin=0.5, ymax=2.5,
width=\figurewidth,
height=\figureheight,
tick align=outside,
xmajorgrids,
x grid style={white},
ymajorgrids,
y grid style={white},
axis line style={white},
axis background/.style={fill=white!89.803921568627459!black}
]
\addplot [blue, mark=*, mark size=3, mark options={solid,draw=black}, only marks]
table {%
1 1
2 1
3 2
4 1
5 1
6 1
7 2
8 1
9 1
10 1
11 2
12 2
13 1
14 1
15 1
16 2
17 1
18 1
19 1
20 2
21 1
22 1
23 1
24 2
25 1
26 1
27 1
28 2
29 1
30 1
31 1
32 2
33 2
34 1
35 1
36 1
37 2
38 1
39 1
40 1
41 2
42 1
43 1
44 1
45 2
46 1
47 1
48 1
49 2
50 1
};
\addplot [color0, dash pattern=on 1pt off 3pt on 3pt off 3pt]
table {%
1 0
1 1
};
\addplot [color1, dash pattern=on 1pt off 3pt on 3pt off 3pt]
table {%
2 0
2 1
};
\addplot [color2, dash pattern=on 1pt off 3pt on 3pt off 3pt]
table {%
3 0
3 2
};
\addplot [white!46.666666666666664!black, dash pattern=on 1pt off 3pt on 3pt off 3pt]
table {%
4 0
4 1
};
\addplot [color3, dash pattern=on 1pt off 3pt on 3pt off 3pt]
table {%
5 0
5 1
};
\addplot [color4, dash pattern=on 1pt off 3pt on 3pt off 3pt]
table {%
6 0
6 1
};
\addplot [color5, dash pattern=on 1pt off 3pt on 3pt off 3pt]
table {%
7 0
7 2
};
\addplot [color0, dash pattern=on 1pt off 3pt on 3pt off 3pt]
table {%
8 0
8 1
};
\addplot [color1, dash pattern=on 1pt off 3pt on 3pt off 3pt]
table {%
9 0
9 1
};
\addplot [color2, dash pattern=on 1pt off 3pt on 3pt off 3pt]
table {%
10 0
10 1
};
\addplot [white!46.666666666666664!black, dash pattern=on 1pt off 3pt on 3pt off 3pt]
table {%
11 0
11 2
};
\addplot [color3, dash pattern=on 1pt off 3pt on 3pt off 3pt]
table {%
12 0
12 2
};
\addplot [color4, dash pattern=on 1pt off 3pt on 3pt off 3pt]
table {%
13 0
13 1
};
\addplot [color5, dash pattern=on 1pt off 3pt on 3pt off 3pt]
table {%
14 0
14 1
};
\addplot [color0, dash pattern=on 1pt off 3pt on 3pt off 3pt]
table {%
15 0
15 1
};
\addplot [color1, dash pattern=on 1pt off 3pt on 3pt off 3pt]
table {%
16 0
16 2
};
\addplot [color2, dash pattern=on 1pt off 3pt on 3pt off 3pt]
table {%
17 0
17 1
};
\addplot [white!46.666666666666664!black, dash pattern=on 1pt off 3pt on 3pt off 3pt]
table {%
18 0
18 1
};
\addplot [color3, dash pattern=on 1pt off 3pt on 3pt off 3pt]
table {%
19 0
19 1
};
\addplot [color4, dash pattern=on 1pt off 3pt on 3pt off 3pt]
table {%
20 0
20 2
};
\addplot [color5, dash pattern=on 1pt off 3pt on 3pt off 3pt]
table {%
21 0
21 1
};
\addplot [color0, dash pattern=on 1pt off 3pt on 3pt off 3pt]
table {%
22 0
22 1
};
\addplot [color1, dash pattern=on 1pt off 3pt on 3pt off 3pt]
table {%
23 0
23 1
};
\addplot [color2, dash pattern=on 1pt off 3pt on 3pt off 3pt]
table {%
24 0
24 2
};
\addplot [white!46.666666666666664!black, dash pattern=on 1pt off 3pt on 3pt off 3pt]
table {%
25 0
25 1
};
\addplot [color3, dash pattern=on 1pt off 3pt on 3pt off 3pt]
table {%
26 0
26 1
};
\addplot [color4, dash pattern=on 1pt off 3pt on 3pt off 3pt]
table {%
27 0
27 1
};
\addplot [color5, dash pattern=on 1pt off 3pt on 3pt off 3pt]
table {%
28 0
28 2
};
\addplot [color0, dash pattern=on 1pt off 3pt on 3pt off 3pt]
table {%
29 0
29 1
};
\addplot [color1, dash pattern=on 1pt off 3pt on 3pt off 3pt]
table {%
30 0
30 1
};
\addplot [color2, dash pattern=on 1pt off 3pt on 3pt off 3pt]
table {%
31 0
31 1
};
\addplot [white!46.666666666666664!black, dash pattern=on 1pt off 3pt on 3pt off 3pt]
table {%
32 0
32 2
};
\addplot [color3, dash pattern=on 1pt off 3pt on 3pt off 3pt]
table {%
33 0
33 2
};
\addplot [color4, dash pattern=on 1pt off 3pt on 3pt off 3pt]
table {%
34 0
34 1
};
\addplot [color5, dash pattern=on 1pt off 3pt on 3pt off 3pt]
table {%
35 0
35 1
};
\addplot [color0, dash pattern=on 1pt off 3pt on 3pt off 3pt]
table {%
36 0
36 1
};
\addplot [color1, dash pattern=on 1pt off 3pt on 3pt off 3pt]
table {%
37 0
37 2
};
\addplot [color2, dash pattern=on 1pt off 3pt on 3pt off 3pt]
table {%
38 0
38 1
};
\addplot [white!46.666666666666664!black, dash pattern=on 1pt off 3pt on 3pt off 3pt]
table {%
39 0
39 1
};
\addplot [color3, dash pattern=on 1pt off 3pt on 3pt off 3pt]
table {%
40 0
40 1
};
\addplot [color4, dash pattern=on 1pt off 3pt on 3pt off 3pt]
table {%
41 0
41 2
};
\addplot [color5, dash pattern=on 1pt off 3pt on 3pt off 3pt]
table {%
42 0
42 1
};
\addplot [color0, dash pattern=on 1pt off 3pt on 3pt off 3pt]
table {%
43 0
43 1
};
\addplot [color1, dash pattern=on 1pt off 3pt on 3pt off 3pt]
table {%
44 0
44 1
};
\addplot [color2, dash pattern=on 1pt off 3pt on 3pt off 3pt]
table {%
45 0
45 2
};
\addplot [white!46.666666666666664!black, dash pattern=on 1pt off 3pt on 3pt off 3pt]
table {%
46 0
46 1
};
\addplot [color3, dash pattern=on 1pt off 3pt on 3pt off 3pt]
table {%
47 0
47 1
};
\addplot [color4, dash pattern=on 1pt off 3pt on 3pt off 3pt]
table {%
48 0
48 1
};
\addplot [color5, dash pattern=on 1pt off 3pt on 3pt off 3pt]
table {%
49 0
49 2
};
\addplot [color0, dash pattern=on 1pt off 3pt on 3pt off 3pt]
table {%
50 0
50 1
};
\addplot [red]
table {%
1 0
50 0
};
\nextgroupplot[
xlabel={c) Time},
ylabel={Sensor $i$ on duty},
xmin=0, xmax=25,
ymin=0.5, ymax=2.5,
width=\figurewidth,
height=\figureheight,
tick align=outside,
xmajorgrids,
x grid style={white},
ymajorgrids,
y grid style={white},
axis line style={white},
axis background/.style={fill=white!89.803921568627459!black}
]
\addplot [blue, mark=*, mark size=3, mark options={solid,draw=black}, only marks]
table {%
1 1
2 1
3 2
4 1
5 1
6 1
7 2
8 1
9 1
10 1
11 1
12 2
13 1
14 1
15 1
16 2
17 1
18 1
19 1
20 2
21 1
22 1
23 1
24 2
25 1
26 1
27 1
28 2
29 1
30 1
31 1
32 1
33 2
34 1
35 1
36 1
37 2
38 1
39 1
40 1
41 2
42 1
43 1
44 1
45 2
46 1
47 1
48 1
49 2
50 1
};
\addplot [color0, dash pattern=on 1pt off 3pt on 3pt off 3pt]
table {%
1 0
1 1
};
\addplot [color1, dash pattern=on 1pt off 3pt on 3pt off 3pt]
table {%
2 0
2 1
};
\addplot [color2, dash pattern=on 1pt off 3pt on 3pt off 3pt]
table {%
3 0
3 2
};
\addplot [white!46.666666666666664!black, dash pattern=on 1pt off 3pt on 3pt off 3pt]
table {%
4 0
4 1
};
\addplot [color3, dash pattern=on 1pt off 3pt on 3pt off 3pt]
table {%
5 0
5 1
};
\addplot [color4, dash pattern=on 1pt off 3pt on 3pt off 3pt]
table {%
6 0
6 1
};
\addplot [color5, dash pattern=on 1pt off 3pt on 3pt off 3pt]
table {%
7 0
7 2
};
\addplot [color0, dash pattern=on 1pt off 3pt on 3pt off 3pt]
table {%
8 0
8 1
};
\addplot [color1, dash pattern=on 1pt off 3pt on 3pt off 3pt]
table {%
9 0
9 1
};
\addplot [color2, dash pattern=on 1pt off 3pt on 3pt off 3pt]
table {%
10 0
10 1
};
\addplot [white!46.666666666666664!black, dash pattern=on 1pt off 3pt on 3pt off 3pt]
table {%
11 0
11 1
};
\addplot [color3, dash pattern=on 1pt off 3pt on 3pt off 3pt]
table {%
12 0
12 2
};
\addplot [color4, dash pattern=on 1pt off 3pt on 3pt off 3pt]
table {%
13 0
13 1
};
\addplot [color5, dash pattern=on 1pt off 3pt on 3pt off 3pt]
table {%
14 0
14 1
};
\addplot [color0, dash pattern=on 1pt off 3pt on 3pt off 3pt]
table {%
15 0
15 1
};
\addplot [color1, dash pattern=on 1pt off 3pt on 3pt off 3pt]
table {%
16 0
16 2
};
\addplot [color2, dash pattern=on 1pt off 3pt on 3pt off 3pt]
table {%
17 0
17 1
};
\addplot [white!46.666666666666664!black, dash pattern=on 1pt off 3pt on 3pt off 3pt]
table {%
18 0
18 1
};
\addplot [color3, dash pattern=on 1pt off 3pt on 3pt off 3pt]
table {%
19 0
19 1
};
\addplot [color4, dash pattern=on 1pt off 3pt on 3pt off 3pt]
table {%
20 0
20 2
};
\addplot [color5, dash pattern=on 1pt off 3pt on 3pt off 3pt]
table {%
21 0
21 1
};
\addplot [color0, dash pattern=on 1pt off 3pt on 3pt off 3pt]
table {%
22 0
22 1
};
\addplot [color1, dash pattern=on 1pt off 3pt on 3pt off 3pt]
table {%
23 0
23 1
};
\addplot [color2, dash pattern=on 1pt off 3pt on 3pt off 3pt]
table {%
24 0
24 2
};
\addplot [white!46.666666666666664!black, dash pattern=on 1pt off 3pt on 3pt off 3pt]
table {%
25 0
25 1
};
\addplot [color3, dash pattern=on 1pt off 3pt on 3pt off 3pt]
table {%
26 0
26 1
};
\addplot [color4, dash pattern=on 1pt off 3pt on 3pt off 3pt]
table {%
27 0
27 1
};
\addplot [color5, dash pattern=on 1pt off 3pt on 3pt off 3pt]
table {%
28 0
28 2
};
\addplot [color0, dash pattern=on 1pt off 3pt on 3pt off 3pt]
table {%
29 0
29 1
};
\addplot [color1, dash pattern=on 1pt off 3pt on 3pt off 3pt]
table {%
30 0
30 1
};
\addplot [color2, dash pattern=on 1pt off 3pt on 3pt off 3pt]
table {%
31 0
31 1
};
\addplot [white!46.666666666666664!black, dash pattern=on 1pt off 3pt on 3pt off 3pt]
table {%
32 0
32 1
};
\addplot [color3, dash pattern=on 1pt off 3pt on 3pt off 3pt]
table {%
33 0
33 2
};
\addplot [color4, dash pattern=on 1pt off 3pt on 3pt off 3pt]
table {%
34 0
34 1
};
\addplot [color5, dash pattern=on 1pt off 3pt on 3pt off 3pt]
table {%
35 0
35 1
};
\addplot [color0, dash pattern=on 1pt off 3pt on 3pt off 3pt]
table {%
36 0
36 1
};
\addplot [color1, dash pattern=on 1pt off 3pt on 3pt off 3pt]
table {%
37 0
37 2
};
\addplot [color2, dash pattern=on 1pt off 3pt on 3pt off 3pt]
table {%
38 0
38 1
};
\addplot [white!46.666666666666664!black, dash pattern=on 1pt off 3pt on 3pt off 3pt]
table {%
39 0
39 1
};
\addplot [color3, dash pattern=on 1pt off 3pt on 3pt off 3pt]
table {%
40 0
40 1
};
\addplot [color4, dash pattern=on 1pt off 3pt on 3pt off 3pt]
table {%
41 0
41 2
};
\addplot [color5, dash pattern=on 1pt off 3pt on 3pt off 3pt]
table {%
42 0
42 1
};
\addplot [color0, dash pattern=on 1pt off 3pt on 3pt off 3pt]
table {%
43 0
43 1
};
\addplot [color1, dash pattern=on 1pt off 3pt on 3pt off 3pt]
table {%
44 0
44 1
};
\addplot [color2, dash pattern=on 1pt off 3pt on 3pt off 3pt]
table {%
45 0
45 2
};
\addplot [white!46.666666666666664!black, dash pattern=on 1pt off 3pt on 3pt off 3pt]
table {%
46 0
46 1
};
\addplot [color3, dash pattern=on 1pt off 3pt on 3pt off 3pt]
table {%
47 0
47 1
};
\addplot [color4, dash pattern=on 1pt off 3pt on 3pt off 3pt]
table {%
48 0
48 1
};
\addplot [color5, dash pattern=on 1pt off 3pt on 3pt off 3pt]
table {%
49 0
49 2
};
\addplot [color0, dash pattern=on 1pt off 3pt on 3pt off 3pt]
table {%
50 0
50 1
};
\addplot [red]
table {%
1 0
50 0
};
\end{groupplot}

\end{tikzpicture}